\documentclass[12pt]{article}
\usepackage{amsmath}
\usepackage{amssymb}
\usepackage[english]{babel}
\usepackage{amsthm}
\usepackage{bm}
\usepackage[latin1]{inputenc}
\usepackage{authblk}
\usepackage{graphicx}
\usepackage{makeidx}
\usepackage[left=2cm, right=2cm, top=3cm, bottom=3.5cm]{geometry}
\DeclareGraphicsExtensions{.eps}
\usepackage{color}
\usepackage{cancel}
\usepackage{marginnote}
\usepackage{enumerate}
\usepackage{faktor}
\usepackage{psfrag}
\usepackage{epsfig}
\usepackage{float}
\usepackage{mathrsfs,amsfonts}
\usepackage{comment}

\newtheorem{theorem}{Theorem}
\newtheorem{lemma}{Lemma}
\newtheorem{prop}{Proposition}

\newtheorem{cor}{Corollary}
\newtheorem{remark}{Remark}

{\left\lbrace\begin{array}{@{}l@{}}}%
{\end{array}\right.}

{\left( \begin{array}{@{}l@{}}}%
{\end{array}\right)}

\def\E{\mathcal{E}}
\def\br{\bm r}
\def\brd{\bm{\dot r}}
\def\brdd{\bm{\ddot r}}
\def\bp{\bm p}
\def\bpd{\bm{\dot p}}
\def\l{\ell}
\def\L{L}
\def\H{\mathcal H}

\def\bI{{\bf I}}
\def\bphi{{\bm \varphi}}
\def\g{\mathtt g}
\def\s{\mathtt s}
\def\n{\mathtt n}
\def\k{\mathtt{k}}
\def\T{\mathbb{T}}
\def\N{\mathbb{N}}
\def\eps{\epsilon}
\def\Z{\mathbb{Z}}

\def\A{\mathcal{A}}
\def\W{\mathcal{W}}

\def\G{\mathcal{G}}
\def\D{\mathcal{D}}
\def\Y{\mathcal{Y}}
\def\X{\mathcal{X}}
\def\bl{\boldsymbol{\ell}}
\def\bg{\boldsymbol{g}}
\def\bz{\boldsymbol{z}}
\def\bL{\boldsymbol{L}}
\def\bG{\boldsymbol{G}}
\def\bZ{\boldsymbol{Z}}

\def\ov{\overline}

\begin{document}
\title{Long term dynamics for the restricted $N$-body problem with
   mean motion resonances and crossing singularities}

\author[1,2]{Stefano Mar\`o} 
\author[2]{Giovanni F. Gronchi}

\affil[1]{Instituto de Ciencias Matem\'aticas (CSIC-UAM-UCM-UC3M), Madrid, Spain, \protect\\ email: stefano.maro@icmat.es}
\affil[2]{Dipartimento di Matematica, Universit\`a di Pisa, Italy, \protect\\  email: giovanni.federico.gronchi@unipi.it}

\date{}                     
\setcounter{Maxaffil}{0}
\renewcommand\Affilfont{\small}

\maketitle

\abstract{We consider the long term dynamics of the restricted
  $N$-body problem, modeling in a statistical sense the motion of an
  asteroid in the gravitational field of the Sun and the solar system
  planets. We deal with the case of a mean motion resonance with one
  planet and assume that the osculating trajectory of the asteroid
  crosses the one of some planet, possibly different from the resonant
  one, during the evolution.  Such crossings produce singularities in
  the differential equations for the motion of the asteroid, obtained
  by standard perturbation theory. In this work we prove that the
  vector field of these equations can be extended to two locally
  Lipschitz-continuous vector fields on both sides of a set of
  crossing conditions. This allows us to define generalized solutions,
  continuous but not differentiable, going beyond these
  singularities. Moreover, we prove that the long term evolution of
  the 'signed' orbit distance (Gronchi and Tommei 2007) between the
  asteroid and the planet is differentiable in a neighborhood of the
  crossing times.  In case of crossings with the resonant planet we
  recover the known dynamical protection mechanism against collisions.
  We conclude with a numerical comparison between the long term and
  the full evolutions in the case of asteroids belonging to the
  'Alinda' and 'Toro' classes (Milani et al. 1989).  This work extends
  the results in (Gronchi and Tardioli 2013) to the relevant case of
  asteroids in mean motion resonance with a planet.}

\section{Introduction}

It is well known that for $N\geq 3$ the $N$-body
problem is not integrable, even in the restricted case.
In particular, the evolutions of near-Earth
asteroids (NEAs) have short Lyapunov times, beyond which the orbit
computed by numerical techniques and the true orbit are completely
uncorrelated \cite{Whipple}.
However, we can obtain statistical information on the
long term evolution by considering a normal form of
the Hamiltonian of the problem, where we try to filter out the short
periodic oscillations. More precisely, we would like to eliminate the
dependence on the fast angles from the first order part of the
Hamiltonian \cite{akn}.  Outside mean motion resonances this program can
be successfully completed and
corresponds to averaging Hamilton's equations over the mean anomalies
of the asteroid and the planets. In case of mean motion resonances,
the resonant combination of the mean anomalies is a slow angle and
must be retained in the normal form.

In both cases, the elimination of the fast angles is usually obtained
through a canonical transformation, in the spirit of classical
perturbation theory. However, the intersections between the
trajectories of the asteroid and the planets introduce
  singularities in the standard procedure. Actually, even the
coefficients of the Fourier series expansion of the generating
function are not defined in a neighborhood of
  crossings. On the other hand, since the trajectory of a near-Earth
asteroid is likely to cross the trajectory of the Earth, we
cannot avoid to deal with these problems.
Note that the minimal distance between the trajectories of an asteroid
and a planet is crucial in the study of possible Earth
impactors. Actually, a small value of this quantity, that we denote by
$d_{min}$, is a necessary condition for an impact. An orbit crossing
singularity occurs whenever $d_{min}=0$.
 
After the preliminary study by Lidov and Ziglin \cite{lidzig74}, in
the case of orbits uniformly close to a circular
one, the problem of averaging over crossing orbits
was studied in \cite{gm98}. Here the authors assumed the orbits of the
planets being circular and coplanar, and excluded mean motion
resonances and close approaches with them.  In \cite{g02CMDA} the
results were extended to the case of non-zero eccentricities and
inclinations. In these works, the main singular term
  is computed through a Taylor expansion centered at the mutual nodes
  of the osculating orbits. These results were improved in
\cite{gt13}, where the main singular term is expanded at the minimum
distance points (see Section~\ref{sec:singularity}) and where it is
proved that the averaged vector field admits two different
Lipschitz-continuous extensions in a neighborhood of almost every
crossing configuration.  The latter property allows us to define a
generalized solution, representing the secular evolution of the
asteroid, that is continuous but not differentiable at crossings.
Moreover, one can suitably choose the sign of $d_{min}$ and obtain a
map $\tilde{d}_{min}$ that is differentiable in a neighborhood of
almost all crossing configurations \cite{gt07}.  The secular evolution
of $\tilde{d}_{min}$ along the generalized solutions turns out to be
differentiable in a neighborhood of the singularity.

The basic model considered in these works comes from the averaging
principle. Therefore, it is assumed that the dynamics is not affected
by mean motion resonances. However, the population of resonant NEAs is
not negligible.
Moreover, mean motion resonances are considered responsible for a
relatively fast change in the orbital elements leading some asteroids
to cross the planet trajectories \cite{wisdom}. Hence it is important
to extend the analysis to such asteroids, which is
  the purpose of this paper.

For the resonant case, the averaging process suffers the presence of
small divisors. Hence, the dependence on the mean anomalies cannot be
completely eliminated, and the terms corresponding to their resonant
combination still appear in the resonant normal form, see
(\ref{normalform}).  We observe that in this relation the averaged
Hamiltonian considered in \cite{gt13} is still present. However, a new
term
appears in the form of a Fourier series, that we truncate to some
order $n_{\rm max}$. This term, denoted by $\H_{res}^{n_{\rm
    max}}$, is singular at orbit crossings and needs to be studied.
Another difference with the non-resonant case is that the semimajor
axis of the asteroid orbit is not constant, and the
number of state variables to consider in the equations is six.  

We will prove that, despite these differences, the vector field of the
resonant normal form computed outside the singularities admits two
different locally Lipschitz-continuous extensions on both sides of a
set of crossing conditions, as in \cite{gt13}. We can also
define generalized solutions, continuous but not differentiable, going
beyond the crossing singularities and
the long term evolution of the map $\tilde{d}_{min}$ along these
solutions is differentiable in a neighborhood of crossings. 

The analysis of the singularity is performed in two different ways,
depending if the crossed planet is the one in mean motion resonance
with the asteroid or not.
In case of crossings with the resonant planet we show that, in the
limit for $n_{\rm max}\to\infty$, we recover the known dynamical
protection mechanism against collisions between the asteroid and the
planet \cite{milbac98}.

The article is organized as follows. In Section \ref{sec:av_eq} we 
derive the equations of the long term dynamics outside the crossing
singularities for a given mean motion resonance.
In Section \ref{sec:distance} we recall the definition of the signed
orbit distance $\tilde{d}_{min}$.  The main results are stated and
proved in Section \ref{sec:singularity}. In Section \ref{sec:gen} we
define the generalized solutions and prove the regularity of the
evolution of $\tilde{d}_{min}$.  In Section~\ref{sec:dynpr} we show
the relation between the resonant normal form that we use and the
averaged Hamiltonian used in the literature, recovering the dynamical
mechanism that protects from collisions. We conclude with some numerical
examples in Section \ref{sec:numerics}, showing the agreement between
the long term evolution and the full evolution in a statistical sense.

\section{The equations for the long term evolution}\label{sec:av_eq}

We consider the differential equations
\begin{equation}
\brdd = -\k^2\frac{\br}{|\br|^3}+\k^2\sum_{j=1}^{N-2} \mu_j\left(
\frac{\br_j-\br}{|\br_j-\br|^3}-\frac{\br_j}{|\br_j|^3} \right),
\label{Nbodyeq}
 \end{equation}
where $\br$ describes, in heliocentric coordinates, the motion of a
massless asteroid under the gravitational attraction of the Sun and
$N-2$ planets. The heliocentric motions of the planets $\br_j =
\br_j(t)$ are known functions of the time $t$ that never vanish: that
is we exclude collisions between a planet and the Sun. Moreover, $\k =
\sqrt{\mathcal{G}m_0}$ is Gauss's constant, $\mu_j= m_j/m_0$ with
$m_0$ the mass of the Sun and $m_j$ the mass of the $j$-th planet.
Equations (\ref{Nbodyeq}) can be written in Hamiltonian form as
\begin{equation*}
\bpd = -\frac{\partial \H}{\partial \br}, \hskip 1cm
\brd = \frac{\partial \H}{\partial \bp} = \bp, 
\end{equation*}
with Hamiltonian
\begin{equation}
\H(\bp,\br,t) = \frac{|\bp |^2}{2}-\frac{\k^2}{|\br|}-\k^2\sum_{j=1}^{N-2}
\mu_j\left( \frac{1}{d_j(\br,t)}-\frac{\br\cdot\br_j(t)}{|\br_j(t)|^3} \right).
\label{hamorig}
\end{equation}
In (\ref{hamorig}) $d_j=|\br_j-\br|$ stands for the distance between
the asteroid and the $j$-th planet.
We use Delaunay's elements $(L,G,Z,l,g,z)$ defined by 
\begin{align*}
\L &=\k\sqrt{ a},  &\l = \n(t-t_0), \\ 
G &=\k\sqrt{a(1-e^2)},      & g =\omega, \\
Z &=\k\sqrt{a(1-e^2)}\cos I,    & z =\Omega,
\end{align*}
where $a,e,I,\Omega,\omega,t_0$ represent semimajor axis,
eccentricity, inclination, longitude of the ascending node, argument
of perihelion, and epoch of passage at perihelion. For the definition
of $\l$ we use the mean motion
\[
\n=\frac{\k^{4}}{\L^3}.
\]
In these coordinates, the Hamiltonian (\ref{hamorig}) can be written as
\begin{equation*}
\H = \H_0 + \eps\H_1,
\label{hamdel}
\end{equation*}
with $\eps =\mu_5$,
\begin{equation*}
\H_0 =  -\frac{\k^4}{2\L^2}, 
\end{equation*}
and
\begin{equation}
\H_1 = \sum_{j=1}^{N-2} \H_1^{(j)},
\hskip 1cm
\H_1^{(j)} = -\k^2\frac{\mu_j}{\mu_5}\left(
\frac{1}{d_j}-\frac{\br\cdot\br_j}{|\br_j|^3} \right),
\label{H1}
\end{equation}
and $\br_j=\br_j(t)$.
Note that in (\ref{H1})
\[
\H_1 = \H_1(\L,G,Z,\l,g,z,t).
\]
To eliminate the dependence on time in $\H_1$
we overextend the phase space. We assume that the planets move on
quasi-periodic orbits with three independent frequencies $\n_j,\g_j,\s_j$.\\ 
This is the case considered by Laplace (see for example \cite{morby}), where the mean
semi-major axis $a_j$ is constant and the mean value of the mean
anomaly $\l_j$ grows linearly with time, i.e. up to a phase,
$\l_j=\n_jt$. 
Here $\n_j$ is
the mean motion of planet
$j$. Moreover, every planet is characterized by two more frequencies
$\g_j,\s_j$, describing the slow motions of the other mean orbital
elements.
 We introduce the angles
\[
\l_j=\n_jt + \ell_j(0), \quad 
g_j=\g_j t + g_j(0), \quad 
z_j=\s_j t + z_j(0)
\]   
and their conjugate variables $L_j, G_j, Z_j$. 

Note that these variables do not correspond to the Delaunay's elements
of planet $j$, since they are functions of the orbital elements of the
asteroid and planet $j$.  We use the following notation:
\begin{eqnarray*}
\bl = (\ell,\ell_1\dots,\ell_N), &\bg = (g,g_1\dots,g_N), & \bz = (z,z_1\dots,z_N), \\
\bl_j = (\ell,\ell_j), & \bg_j = (g,g_j), & \bz_j = (z,z_j)
\end{eqnarray*}
and analogously we define $\bL,\bG,\bZ,\bL_j,\bG_j,\bZ_j$.

The dynamics in this overextended phase space is determined by the
autonomous Hamiltonian
\begin{equation*}
\widetilde{\H} = -\frac{\k^4}{2\L^2} + \sum_{j=1}^{N-2}(\n_j\L_j+ \g_jG_j+
\s_jZ_j) + \eps\widetilde{\H}_1(\L,G,Z,\bl,\bg,\bz),
\end{equation*}
where
\[
\widetilde{\H}_1 = \sum_{j=1}^{N-2}\tilde{\H_1}^{(j)}, 
\hskip 1cm
\widetilde{\H}_1^{(j)}= -\k^2\frac{\mu_j}{\mu_5} \left( \frac{1}{\tilde{d}_j}-\frac{\br\cdot\tilde{\br}_j}{|\tilde{\br}_j|^3}
\right),
\]
with
\[
\tilde{\br}_j =\tilde{\br}_j(\l_j,g_j,z_j), \hskip 0.6cm
\tilde{d}_j = |\tilde{\br}_j - \br|.
\] 
Here we are assuming that $\br_j$ evolves according to Laplace's
solution for the planetary motions, and we write it as a function of
its frequencies, denoted by $\tilde{\br}_j$.
Hereafter we shall omit the 'tilde', to simplify the notation.\\
The frequencies $\g_j$ and $\s_j$ are of order $\epsilon$
\cite{morby}.  In order to study the secular dynamics, we would like
to eliminate all the frequencies corresponding to the fast angles
$\bl$. In case of a mean motion resonance with a planet this is not
possible.

In the following we shall assume that there is only one mean motion
resonance with a planet and no close approaches occur. To expose our
result we shall consider a $|h_5^*|:|h^*|$ mean motion resonance with
Jupiter given by
\begin{equation}
h^* \n + h^*_5 \n_5 =0 \quad\mbox{for some }(h^*,h^*_5)\in\mathbb{Z}^2.
\label{Jup_res}
\end{equation}
A mean motion resonance with another planet can be treated in a similar way.
We denote by
\[
\bphi=(\bl,\bg,\bz),\quad \bphi_j=(\bl_j,\bg_j,\bz_j)
\]
the vectors of the angles and by
\[
\bI=(\bL,\bG,\bZ),\quad \bI_j=(\bL_j,\bG_j,\bZ_j)
\]
the corresponding vectors of the actions.  \\
We use the Lie method \cite{morby} to search for a suitable canonical
transformation close to the identity, that is we search for a function
$\chi=\chi(\bI',\bphi')$ such that the inverse transformation is
\begin{equation*}
\Phi^\eps_\chi(\bI',\bphi') = (\bI,\bphi),
\end{equation*}
where $\Phi^t_\chi$ is the Hamiltonian flow associated to $\chi$.
The function $\chi$ is selected so that the transformed Hamiltonian 
$\H^\prime = \H\circ\Phi^\eps_\chi$ depends, at least at first order, on as less fast angular variables as
possible.
Using a formal expansion in $\eps$ we have
\begin{equation*}
\H' = \H\circ\Phi^\eps_\chi = \H + \eps\{\H,\chi\} + O(\eps^2) = 
\H_0 + \eps(\H_1 + \{\H_0,\chi\}) + O(\eps^2).
\end{equation*}
In the resonant case we search for a solution $\chi$ of the equation
\begin{equation}
\H_1 + \{\H_0,\chi\} = f
\label{homologic}
\end{equation}
for some function $f= f(\bI',h^*\l'+h_5^*\l_5',\bg',\bz')$.
To solve (\ref{homologic}) we restrict to the case where no orbit
crossings with the planets occur. We shall see in the next sections how we
can deal with the case of crossings.

We develop 
\[
\H_1  = \sum_{j=1}^{N-2} \H_1^{(j)} 
\]
in Fourier's series of the fast angles:
\begin{equation*}
\H_1^{(j)}= \sum_{(h,h_{j})\in\Z^2}\widehat{\H}_{(h,h_{j})}^{(j)} e^{i(h\l+h_{j}\l_j)}.
\end{equation*}
Here 
\begin{equation}
\widehat{\H}_{(h,h_{j})}^{(j)} = \widehat{\H}_{(h,h_{j})}^{(j)}(L,G,Z,\bg_j,\bz_j ) = \frac{1}{(2\pi)^2}\int_{\T^2}\H_1^{(j)}
e^{ -i(h\l+h_{j}\l_j)}d\l d\l_j
\label{fouriercoe}
\end{equation}
are the Fourier coefficients.  We observe that $\widehat{\H}_{(h,h_{j})}^{(j)}$ are
defined also in case of orbit crossings, since the integral in
(\ref{fouriercoe}) converges (see e.g. \cite{gt13}).
\\
Moreover, we can write $\chi$ as
\begin{equation*}
\chi = \sum_{j=1}^{N-2} \chi^{(j)}, \hskip 1cm
\chi^{(j)}=\chi^{(j)}(\L^\prime, G^\prime, Z^\prime,\bl_j^\prime, \bg_j^\prime, \bz_j^\prime)
\end{equation*}
and search for the coefficients
\[
\widehat{\chi}_{(h,h_{j})}^{(j)} = \widehat{\chi}_{(h,h_{j})}^{(j)}(L',G',Z',\bg_j',\bz_j' )
\]
in the Fourier series development 
\[
\chi^{(j)}= \sum_{(h,h_{j})\in\Z^2}\widehat{\chi}_{(h,h_{j})}^{(j)} e^{i(h\l^\prime+h_{j}\l_j^\prime )}.
\] 
Inserting these Fourier developments into (\ref{homologic})
we obtain
\[
\H_1 + \{\H_0,\chi\} = \sum_{j=1}^{N-2} 
\Bigl(\H_1^{(j)} - \frac{\partial\H_0}{\partial \bI}\cdot\frac{\partial\chi^{(j)}}{\partial\bphi}\Bigr),
\]
where
\begin{eqnarray*}
&& \H_1^{(j)} - \frac{\partial\H_0}{\partial \bI}\cdot\frac{\partial\chi^{(j)}}{\partial\bphi} = 
\sum_{(h,h_{j})\in\Z^2} \bigl[ \widehat{\H}_{(h,h_{j})}^{(j)} - i (h\n + h_{j}\n_j)\widehat{\chi}_{(h,h_{j})}^{(j)}\bigr]
  e^{i(h\l^\prime+h_{j}\l_j^\prime)}.
\end{eqnarray*}
This expression suggests to choose the function $f$ in
(\ref{homologic}) in the following form:
\[
f=\sum_{j=1}^{N-2} f_j,
\] 
where $f_5=f_5(\bI'_5,h^*\l'+h_5^*\l_5',\bg'_5,\bz'_5)$ and
$f_j=f_j(\bI'_j,\bg'_j,\bz'_j)$ for $j\neq 5$.
This can be accomplished by choosing
\begin{equation*}
\widehat{\chi}_{(h,h_{j})}^{(j)} =
\frac{\widehat{\H}_{(h,h_{j})}^{(j)}}{i (h\n + h_{j}\n_j)} 
\end{equation*}
when the denominator does not vanish. Hence, we exclude
the case $(h,h_{j})=(0,0)$ and the resonant case
$(h,h_5)=n(h^*,h_5^*)$ for some $n\in\mathbb{Z}^*=\mathbb{Z}\setminus\{0\}$, 
for which we assume that the corresponding Fourier coefficient of $\chi$
vanishes.
With this choice we have
\begin{eqnarray*}
f_5 &=& \widehat{\H}_{(0,0)}^{(5)} + \sum_{n\in\Z^*}\widehat{\H}_{n(h^*,h_5^*)}^{(5)}e^{in(h^*\l^\prime+ h_5^*\l_5^\prime)},\\
f_j &=& \widehat{\H}_{(0,0)}^{(j)} \qquad \mbox{for }j\neq 5 .
\end{eqnarray*}
We truncate the Fourier series to some order $n_{\rm max}$ and consider
\begin{equation}
\mathscr{H}_{n_{\rm max}} = 
\H_0 + \eps (\overline{\H}_1+\H_{res}^{n_{\rm max}})
\label{normalform}
\end{equation}
as resonant normal form  of the Hamiltonian, where   
\[
\overline{\H}_1 = \sum_{j=0}^{N-2} \widehat{\H}_{(0,0)}^{(j)},
\]
and
\[
\H_{res}^{n_{\rm max}}=  \sum_{1\leq |n|\leq{n_{\rm max}}}\widehat{\H}_{n(h^*,h_5^*)}^{(5)}e^{in(h^*\l^\prime+
  h_5^*\l_5^\prime)} = 2\Re \left(\sum_{n=1}^{n_{\rm max}}\widehat{\H}_{n(h^*,h_5^*)}^{(5)}e^{in(h^*\l^\prime+
  h_5^*\l_5^\prime)} \right),
\]
with $\Re(z)$ the real part of $z\in\mathbb{C}$, where we used
$\overline {\widehat{\H}_{(h,h_5)}^{(5)}} = \widehat{\H}_{(-h,-h_5)}^{(5)}$.
For simplicity, we shall write $\mathscr{H}$, $\H_{res}$ in place of
$\mathscr{H}_{n_{\rm max}}$, $\H_{res}^{n_{\rm max}}$.
It is easy to see that, for every $j$,
\begin{eqnarray*}
\widehat{\H}_{(0,0)}^{(j)} &=& \frac{1}{(2\pi)^2}\int_{\T^2}\H_1^{(j)}d\l d\l_j
= -\frac{\k^2\mu_j}{(2\pi)^2\mu_5}\int_{\T^2}\left(
\frac{1}{d_j}-\frac{\br\cdot\br_j}{|\br_j|^3} \right)d\l d\l_j = \\
&=&-\frac{\k^2\mu_j}{(2\pi)^2\mu_5}\int_{\T^2} \frac{1}{d_j} d\l d\l_j,
\end{eqnarray*}
being null the average of the indirect perturbation (see \cite{g02}).
We observe that in the Fourier coefficient $\widehat{\H}_{n(h^*,h_5^*)}^{(5)}$ the
term corresponding to the indirect perturbation does not vanish.
We can write 
\begin{eqnarray*}
\overline{\H}_1 &=& \sum_{j=0}^{N-2} \frac{C_j}{(2\pi)^2}\int_{\T^2} \frac{1}{d_j}
d\l d\l_j, \\ 
\H_{res}&=& \frac{2C_5}{(2\pi)^2}\sum_{n=1}^{n_{\rm max}}\left[
  I_5^{c,n}\cos n(h^*\l+h^*_5\l_5) + I_5^{s,n}\sin n(h^*\l+h^*_5\l_5) \right],
\end{eqnarray*}
where
\begin{eqnarray*}
C_j &=& -\frac{\k^2\mu_j}{\mu_5} = -\frac{\k^2m_j}{m_5},\\
I_5^{c,n} &=& \int_{\T^2}
\left(\frac{1}{d_5}-\frac{\br\cdot\br_5}{|\br_5|^3} \right)
\cos n(h^*\l+h^*_5\l_5)d\l d\l_5,  \\
I_5^{s,n} &=& \int_{\T^2}
\left(\frac{1}{d_5}-\frac{\br\cdot\br_5}{|\br_5|^3} \right)
\sin n(h^*\l+h^*_5\l_5)d\l d\l_5,
\end{eqnarray*}
with $I_5^{c,n}, I_5^{s,n}$ depend on $L,G,Z,\bg_5,\bz_5$.

Moreover, since the new Hamiltonian does not depend on $\l_j$ for
 $j\neq 5$ we have
\[
\H_0(\L,\L_5,G_1,\dots,G_N,Z_1,\dots,Z_N) = -\frac{\k^4}{2\L^2} +
  \n_5\L_5+ \sum_{j=1}^{N-2}(\g_jG_j+ \s_jZ_j) .
\]
%
 
We now introduce the resonant angle $\sigma$ through the canonical
transformation
\begin{equation*}
\left(
\begin{array}{l}
\sigma \\ 
\sigma_5
\end{array}
\right)
=
A
\left(
\begin{array}{l}
\ell \\ 
\ell_5
\end{array}
\right),
\quad
\left(
\begin{array}{l}
S \\ 
S_5
\end{array}
\right)
=
A^{-T}
\left(
\begin{array}{l}
L \\ 
L_5
\end{array}
\right),
\end{equation*}
with
\begin{equation*}
A=
\left(
\begin{array}{cc}
h^*  &   h_5^* \\ 
0  &   1/h^*
\end{array}
\right),
\hskip 1cm
A^{-T}=
\left(
\begin{array}{cc}
  1/h^*  &0\\
  -h_5^* &h^*
\end{array}
\right).
\end{equation*}
We chose the matrix $A$ so that $L$ does not depend on $S_5$. For this
reason we could not use a unimodular matrix. However, this will not
affect our analysis.

We shall still denote by
\begin{equation}
\mathscr{H}= \H_0 + \eps (\overline{\H}_1+\H_{res}),
\label{hfin}
\end{equation}
the resonant normal form of the Hamiltonian in these new variables,
with
\begin{align*}
&\H_0(S,S_5,G_1,\dots,G_N,Z_1,\dots,Z_N) = -\frac{\k^4}{2(h^* S)^2} +
  \n_5 (h^*_5 S + S_5/h^*)+ \sum_{j=1}^{N-2}(\g_jG_j+ \s_jZ_j), \\
&\H_{res}(S,G,Z,\sigma,{\bg_5,\bz_5}) = \frac{2C_5}{(2\pi)^2} \sum_{n=1}^{n_{\rm max}}(I_5^{c,n} \cos n\sigma
+ I_5^{s,n}\sin n\sigma),
\\
& \overline{\H}_1(S,G,Z,{\bg,\bz})
=\sum_{j=1}^{N-2}\frac{C_j}{(2\pi)^2}\int_{\T^2}\frac{1}{d_j(\ell,\ell_j)}
d\ell d\ell_j.
\end{align*}
Since the Hamiltonian does not depend on $\sigma_5$, the value of
$S_5$ will remain constant and we will treat it as a parameter.
Calling $\Y=(S,G,Z,\sigma,g,z)$ we consider the equations for the motion of the asteroid given by
\begin{equation}\label{hamres}
\dot{\Y} = \mathbb{J}_3\nabla_\Y\mathscr{H},
\end{equation}
where
\begin{equation*}
\mathbb{J}_3=
\left(
\begin{array}{cc}
0 & -I \\
\phantom{-}I & 0
\end{array}
\right) 
\end{equation*}
is the symplectic identity of order $6$.
In components, system (\ref{hamres}) is written as
\begin{align*} 
\dot S &= -\frac{\partial\mathscr{H}}{\partial \sigma} =
-\eps \frac{\partial\H_{res}}{\partial\sigma}, \\
%
\dot \sigma &= \frac{\partial\mathscr{H}}{\partial S}= \frac{h^*\k^4}{(h^* S)^3} +
\n_5 h^*_5 + \eps\left(\frac{\partial \H_{res}}{\partial S}
+\frac{\partial\overline{\H}_1}{\partial
  S}\right),\\
%
\dot G &= -\frac{\partial\mathscr{H}}{\partial g}= -\eps\left(\frac{\partial
  \H_{res}}{\partial
  g}+\frac{\partial\overline{\H}_1}{\partial
  g}\right),\\
%
\dot g &= \frac{\partial\mathscr{H}}{\partial G}= \eps\left(\frac{\partial
  \H_{res}}{\partial G}
+\frac{\partial\overline{\H}_1}{\partial
  G}\right),\\
%
\dot Z &= -\frac{\partial\mathscr{H}}{\partial z}= -\eps\left(\frac{\partial
  \H_{res}}{\partial
  z}+\frac{\partial\overline{\H}_1}{\partial
  z}\right),\\
%
\dot z &= \frac{\partial\mathscr{H}}{\partial Z}= \eps\left(\frac{\partial
  \H_{res}}{\partial Z}
+\frac{\partial\overline{\H}_1}{\partial
  Z}\right).
\end{align*}
where $\H_{res}$, $\overline{\H}_1$ are functions of
$(S,G,Z,\sigma,{\bg_5,\bz_5})$ and $(S,G,Z,{\bg,\bz})$
respectively.
Since $\eps C_j = -\k^2 \mu_j$, we get
\small
\begin{align*} 
\dot S &= \frac{\k^2}{(2\pi)^2} 2\mu_5 \sum_{n=1}^{n_{\rm max}}n\left( I_5^{s,n} \cos  n\sigma  - I_5^{c,n} \sin n\sigma \right),  
\\
\dot \sigma &= \frac{h^*\k^4}{(h^* S)^3} + \n_5 h^*_5  
\\
&- \frac{\k^2}{(2\pi)^2} \biggl\{
2\mu_5 \sum_{n=1}^{n_{\rm max}}\left(\frac{\partial I_5^{c,n}}{\partial S} \cos n\sigma + \frac{\partial I_5^{s,n}}{\partial S}\sin n\sigma\right)
+\sum_{j=1}^{N-2} \mu_j\frac{\partial}{\partial S}\int_{\T^2}\frac{1}{d_j} d\l
d\l_j\biggr\}, 
\\
\dot G &= \frac{\k^2}{(2\pi)^2} 
 \biggl\{
2\mu_5\sum_{n=1}^{n_{\rm max}}\left(\frac{\partial I_5^{c,n}}{\partial g} \cos n\sigma + \frac{\partial I_5^{s,n}}{\partial g}\sin n\sigma\right)
+\sum_{j=1}^{N-2}\mu_j\frac{\partial}{\partial g}\int_{\T^2}\frac{1}{d_j} d\l
d\l_j\biggr\}, 
\\
\dot g &= -\frac{\k^2}{(2\pi)^2}\biggl\{
2\mu_5 \sum_{n=1}^{n_{\rm max}}\left(\frac{\partial I_5^{c,n}}{\partial G} \cos n\sigma + \frac{\partial I_5^{s,n}}{\partial G}\sin n\sigma\right)
+\sum_{j=1}^{N-2}\mu_j\frac{\partial}{\partial G}\int_{\T^2}\frac{1}{d_j} d\l
d\l_j\biggr\}, 
\\
\dot Z &=\frac{ \k^2}{(2\pi)^2}   \biggl\{
2\mu_5\sum_{n=1}^{n_{\rm max}}\left(\frac{\partial I_5^{c,n}}{\partial z} \cos n\sigma + \frac{\partial I_5^{s,n}}{\partial z}\sin n\sigma\right)
+\sum_{j=1}^{N-2}\mu_j\frac{\partial}{\partial z}\int_{\T^2}\frac{1}{d_j} d\l
d\l_j\biggr\}, 
\\
\dot z &= -\frac{\k^2}{(2\pi)^2} \biggl\{ 
2\mu_5 \sum_{n=1}^{n_{\rm max}}\left(\frac{\partial I_5^{c,n}}{\partial Z} \cos n\sigma + \frac{\partial I_5^{s,n}}{\partial Z}\sin n\sigma\right)
+\sum_{j=1}^{N-2} \mu_j\frac{\partial}{\partial Z}\int_{\T^2}\frac{1}{d_j} d\l
d\l_j\biggr\}.
\end{align*}
\normalsize The derivatives of $\H_{res}$ and $\overline{\H}_1$ are
not defined at orbit crossings with the planets. In the following
sections we shall discuss how we can define generalized solutions of
system (\ref{hamres}) in case of orbit crossings.

\section{The orbit distance}
\label{sec:distance}
We recall here some facts and notations from \cite{gt07},
\cite{gt13}. Let $(E,v)$, $(E',v')$ be two sets of orbital elements,
where $E,E'$ describe the trajectories of the asteroid and one planet,
$v,v'$ describe the position of these bodies along them.  Denote by
$\mu'$ the ratio of the mass of this planet to the mass of the Sun.
We also introduce the notation $\mathcal{E}=(E,E')$ for the two-orbit
configuration and $V=(v,v')$ for the vector of parameters along the
orbits.
We denote by $\X=\X(E,v)$ and $\X'=\X'(E',v')$ the Cartesian
coordinates of the asteroid and the planet respectively.  For each
given $\E$, $V_h(\E)$ represents a local minimum point of the
function
\[
V \mapsto d^2(\E,V) = |\X(E,v)-\X'(E',v')|^2.
\]
We introduce
the local minimum maps
\[
\E\mapsto d_h(\E)=d(\E,V_h),
\]
and the orbit distance
\[
\E\mapsto d_{min}(\E)=\min_h d_h(\E).
\] 
We shall consider non-degenerate configurations $\E$, i.e such that
all the critical points of the map $V \mapsto d(\E,V)$ are non-degenerate. In this way, we can always choose a neighborhood $\W$ of
$\E$ where the maps $d_h$ do not have bifurcations.  A crossing
configuration is a two-orbit configuration $\E_c$ such that
$d(\E_c,V_h(\E_c))=0$ where $V_h(\E_c)$ is the corresponding minimum
point.
The maps $d_h$ and $d_{min}$ are singular at crossing configurations,
and their derivatives in general do not exists. Anyway, it is possible
to obtain analytic maps in a neighborhood of a crossing configuration
$\E_c$ by a suitable choice of the sign for these maps. We summarize
here the procedure to deal with this singularity for $d_h$; the
procedure for $d_{min}$ is the same.
Let $V_h=(v_h,v_h')$ be a local minimum point of
$d^2$ and let $\X_h=\X_h(E,v_h)$ and $\X_h'=\X_h'(E',v_h')$.  We
introduce the vectors tangent to the trajectories defined by $E,E'$ at these
points
\[
\tau_h=\frac{\partial\X}{\partial v}(E,v_h),\quad
\tau'_h=\frac{\partial\X'}{\partial v'}(E',v'_h)
\]
and their cross product $\tau_h^*=\tau_h'\times\tau_h$. Both vectors
$\tau_h,\tau'_h$ are orthogonal to $\Delta_h=\X'_h-\X_h$, so
that $\tau_h^*$ is parallel to $\Delta_h$, see Figure
\ref{regrule}.

\begin{figure}[h!]
\psfragscanon
\psfrag{t1}{$\tau_h'$}\psfrag{t2}{$\tau_h$}\psfrag{t3}{$\tau_h^*$}
\psfrag{Dmin}{$\Delta_h$}
\psfrag{Planet orbit}{\small planet orbit}
\psfrag{Asteroid orbit}{\small asteroid orbit}
\centerline{\epsfig{figure=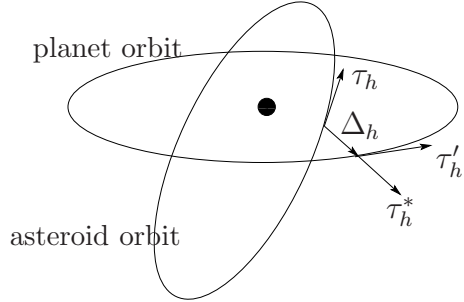,width=6cm}}
\psfragscanoff
\caption{The vectors $\tau_h^*,\Delta_h$.}
\label{regrule}
\end{figure}
\noindent Denoting by $\hat\tau_h^*$, $\hat\Delta_h$ the
corresponding unit vectors, we consider the local minimal distance with sign 
\begin{equation}
\tilde{d}_h=(\hat\tau_h^* \cdot\hat\Delta_h)d_h.
\label{dhtilde}
\end{equation}
This map is analytic in a neighborhood of most crossing
configurations. Actually, this smoothing procedure fails in case the
vectors $\tau_h,\tau'_h$ are parallel.\\
Finally, given a neighborhood $\W$ of $\E_c$ without bifurcations of $d_h$, we write
$\W=\W^-\cup\Sigma\cup\W^+$, where
\[
\Sigma = \W\cap\{\tilde{d}_h(\E)=0\},\quad \W^+=\W\cap\{\tilde{d}_h(\E)>0\},\quad
\W^-=\W\cap\{\tilde{d}_h(\E)<0\}.
\]

\section{Extraction of the singularities}
\label{sec:singularity}

In the following we shall expose a method to investigate the crossing
singularities occurring in \eqref{hamres}. For simplicity, we shall
eventually drop the index 5, referring to Jupiter, and denote simply
by a prime the quantities referring to the crossed planet.

Let $\E_c$ be a two-orbit crossing configuration and suppose that the
trajectories are described by the vector $E=(S,G,Z,g,z)$. In the
following we shall write $y_i$ for the components of the vector
$E$. We choose the mean anomalies as parameters along the trajectory
so that $V=(\ell,\ell')$. The first step of our analysis is to
consider, for each $\E$ in a neighborhood $\W$ of $\E_c$, the Taylor
expansion of $V\mapsto d(\E,V)$ in a neighborhood of $V_h=V_h(\E)$,
i.e.
\[
d^2(\E,V)=d_h^2(\E)+(V-V_h)\cdot\A_h (V-V_h)+{\cal R}_3^{(h)}(\E,V),
\]
where ${\cal R}_3^{(h)}$ is the remainder in the integral form,
and define the approximated distance
\begin{equation}\label{approxdist}
\delta_h(\E,V)=\sqrt{d^2_h(\E) + (V-V_h)\cdot\A_h (V-V_h)},
\end{equation}
with
\[
\A_h=
\left[
\begin{array}{cc}
|\tau_h|^2+\frac{\partial^2\X}{\partial v^2}(E,v_h)\cdot\Delta_h & -\tau_h\cdot\tau_h' \\
-\tau_h\cdot\tau_h' &  |\tau'_h|^2+\frac{\partial^2\X'}{\partial v'^2}(E',v'_h)\cdot\Delta_h .
\end{array}
\right]
\]
The matrix ${\cal A}_h$ is positive definite except for tangent
crossings, where it is degenerate.
To study the crossing singularities in case of a mean motion resonance
with Jupiter we distinguish between the case where the asteroid
trajectory crosses the trajectory of another planet and the case where
it crosses the trajectory of Jupiter itself.
In the first case the crossing singularity appears only in the
averaged terms $\frac{\partial\overline{\H}_1}{\partial y_i}$.
In the second case also the derivatives $\frac{\partial
  I_5^{s,n}}{\partial y_i}$, $\frac{\partial I_5^{c,n}}{\partial y_i}$
are affected by this singularity.
In both cases the component $\frac{\partial \mathscr{H}}{\partial
  \sigma}$ is regular.

We obtain the following results.


\begin{theorem}\label{gttheorem}
  Let $\mathcal{E}_c$ be a non-degenerate crossing configuration
  with a planet (including Jupiter).  Then, there exists a
  neighborhood $\mathcal{W}$ of $\mathcal{E}_c$ such that
  for each $i=1,\ldots, 5$ we can define two maps
$$
{\cal W} \ni \E\mapsto \eps\left(\frac{\partial\overline{\H}_1}{\partial y_i}\right)_h^\pm(\E)
$$
that are  Lipschitz-continuous extensions of the maps
 
\[
  {\cal W}^\pm \ni \E\mapsto
\frac{\mu'\k^2}{(2\pi)^2}\frac{\partial}{\partial
  y_i}\int_{\T^2}\frac{1}{d(\E,V)} dV.
\]
Moreover, the following relation holds in $\W$:
\[
%
\hskip 1.5cm \eps\left(\frac{\partial\overline{\H}_1}{\partial y_i}\right)_h^- -
\eps\left(\frac{\partial\overline{\H}_1}{\partial y_i}\right)_h^+ = -\frac{\mu'
  \k^2}{\pi}\left[\frac{\partial}{\partial
    y_i}\left(\frac{1}{\sqrt{\det(\mathcal{A}_h)}}\right)\tilde{d}_h +
  \frac{1}{\sqrt{\det(\mathcal{A}_h)}}\frac{\partial \tilde{d}_h}{\partial
    y_i} \right].
\]
\end{theorem}
%
\begin{proof}
  We can show this result by following the same steps as in
  \cite[Theorem 4.2]{gt13}, replacing $R$ by $-\eps\overline{\H}_1$.
  
  \end{proof}

\begin{theorem}\label{teosalto}
  Let ${\sf h}=(h^*,h^*_5)$ and $\mathcal{E}_c$ be a non-degenerate crossing
  configuration with Jupiter.
 Then, there exists a
neighborhood $\mathcal{W}$ of $\mathcal{E}_c$ such that, for every $n>0$ 
and for each $i=1,\ldots, 5$, we can define four maps
\[
\mathcal{W}\ni\E\mapsto\left(\frac{\partial I_5^{c,n}}{\partial y_i}\right)_h^\pm(\E),
\qquad
\mathcal{W}\ni\E\mapsto \left(\frac{\partial I_5^{s,n}}{\partial y_i}\right)_h^\pm(\E),
\]
that are   Lipschitz-continuous extensions of the maps

\begin{align} 
  {\cal W}^\pm \ni \E\mapsto &
  \frac{\partial }{\partial y_i}\int_{\T^2}\left(\frac{1}{d(\E,V)}-\frac{\br\cdot\br_5}{|\br_5|^3} \right)\cos(n{\sf h}\cdot V)dV, 
\label{dercos}\\
  {\cal W}^\pm \ni \E\mapsto &
  \frac{\partial }{\partial y_i}\int_{\T^2}\left(\frac{1}{d(\E,V)}-\frac{\br\cdot\br_5}{|\br_5|^3} \right)\sin(n{\sf h}\cdot V)dV,
\label{dersin}
\end{align}
respectively.
Moreover, 
%
%
the following
relations hold in $\W$:
\begin{align*}
  \left(\frac{\partial I_5^{c,n}}{\partial y_i}\right)_h^- -
  \left(\frac{\partial I_5^{c,n}}{\partial y_i}\right)_h^+ &=
  4\pi\cos(n{\sf h}\cdot V_h)\left[\frac{\partial}{\partial y_i}\left(\frac{1}{\sqrt{\det(\mathcal{A}_h)}}\right)\tilde{d}_h+\frac{1}{\sqrt{\det(\mathcal{A}_h)}}\frac{\partial \tilde{d}_h}{\partial y_i}      \right], \\
  \left(\frac{\partial I_5^{s,n}}{\partial y_i}\right)_h^-
  -\left(\frac{\partial I_5^{s,n}}{\partial y_i}\right)_h^+ &=
  4\pi\sin(n{\sf h}\cdot V_h)\left[\frac{\partial}{\partial y_i}\left(\frac{1}{\sqrt{\det(\mathcal{A}_h)}}\right)\tilde{d}_h+\frac{1}{\sqrt{\det(\mathcal{A}_h)}}\frac{\partial \tilde{d}_h}{\partial y_i}       \right].
\end{align*}
\end{theorem}
Before giving a proof of Theorem~\ref{teosalto} we state some
consequences of both theorems.
We define the following locally Lipschitz-continuous maps, extending
the vector field of Hamilton's equations (\ref{hamres}) in a
neighborhood of the crossing singularity,
\[
  {\W\times \T\ni(\E,\sigma)\mapsto}
  \left(\frac{\partial \mathscr{H}}{\partial y_i}\right)_h^\pm{(\E,\sigma)} :=
  \left\{
  \begin{array}{ll}
&\frac{\partial\H_0}{\partial y_i}{(\E)} +
\eps\left(\frac{\partial\overline{\H}_1}{\partial y_i}\right)_h^\pm{(\E)} +
\eps \left(\frac{\partial\H_{res}}{\partial y_i}\right)_h^\pm{(\E,\sigma)},\cr
&\cr
&\frac{\partial\H_0}{\partial y_i}{(\E)} +
\eps\left(\frac{\partial\overline{\H}_1}{\partial y_i}\right)_h^\pm{(\E)},\cr
\end{array}
\right.
\]
where we use the definition above in case of crossings with Jupiter,
and the one below for crossings with other planets.
Here $\H_0$, $\H_{res}$ are defined as in (\ref{hfin}),
and
\[
\eps \left(\frac{\partial\H_{res}}{\partial y_i}\right)_h^\pm{(\E,\sigma)} =
-\frac{2\mu' {\rm k}^2}{(2\pi)^2} \sum_{n=1}^{n_{\rm max}}\Biggl(
\left(\frac{\partial I_5^{c,n}}{\partial y_i}\right)_h^\pm{(\E)} \cos (n\sigma) + \left(\frac{\partial I_5^{s,n}}{\partial
  y_i}\right)_h^\pm{(\E)} \sin(n\sigma) \Biggr).
\]
Moreover, we consider the map 
\[
\W\times \T \ni (\E,\sigma)\mapsto \mathrm{Diff}_h\left(\frac{\partial \mathscr{H}}{\partial y_i}\right)(\E,\sigma) := \left(\frac{\partial \mathscr{H}}{\partial y_i}\right)_h^-{(\E,\sigma)} -
  \left(\frac{\partial \mathscr{H}}{\partial y_i}\right)_h^+{(\E,\sigma)}.
\]

\begin{cor}
\label{cor1}
If $\mathcal{E}_c$ corresponds to a crossing configuration with a
planet different from Jupiter, then the following relation holds in
$\W$:
\begin{align*}
  \mathrm{Diff}_h\left(\frac{\partial \mathscr{H}}{\partial y_i}\right) & = 
\eps\left(\frac{\partial\overline{\H}_1}{\partial y_i}\right)_h^- -
\eps\left(\frac{\partial\overline{\H}_1}{\partial y_i}\right)_h^+ \\  
&= -\frac{\mu'
  \k^2}{\pi}\left[\frac{\partial}{\partial
    y_i}\left(\frac{1}{\sqrt{\det(\mathcal{A}_h)}}\right)\tilde{d}_h +
  \frac{1}{\sqrt{\det(\mathcal{A}_h)}}\frac{\partial \tilde{d}_h}{\partial
    y_i} \right].
  %
\end{align*} 
\end{cor}
\begin{cor}
\label{cor2}
If $\mathcal{E}_c$ corresponds to a crossing configuration with
Jupiter, then the following relation holds in $\W$
\begin{multline*}
  \mathrm{Diff}_h\left(\frac{\partial \mathscr{H}}{\partial y_i}\right) = 
\eps\left(\frac{\partial\overline{\H}_1}{\partial y_i}\right)_h^- -
\eps\left(\frac{\partial\overline{\H}_1}{\partial y_i}\right)_h^+ +
\eps\left(\frac{\partial\H_{res}}{\partial y_i}\right)_h^- -
\eps\left(\frac{\partial\H_{res}}{\partial y_i}\right)_h^+=\\
  %
  %
-\frac{2\mu' \k^2}{\pi} \left[ \sum_{n=1}^{n_{\rm max}}\cos\big(n (\sigma -{\sf h}\cdot V_h)\big)+
  \frac{1}{2}\right]\left[\frac{\partial}{\partial
    y_i}\left(\frac{1}{\sqrt{\det(\mathcal{A}_h)}}\right)\tilde{d}_h+
  \frac{1}{\sqrt{\det(\mathcal{A}_h)}}\frac{\partial
    \tilde{d}_h}{\partial y_i} \right].
\end{multline*}
%
\end{cor}

\medbreak
We recall that, for each $N\in\N$ and $x\neq 2h\pi$, with $h\in\Z$,
we have
\begin{equation}
\sum_{n=1}^N\cos(nx) =
\frac{1}{2}(D_N(x) - 1)
\end{equation}
where
\[
D_N(x) = \frac{\sin\bigl((N+{1}/{2})x\bigr)}{\sin({x}/{2})}
\]
is the Dirichlet kernel (see \cite{stein}). 

\begin{remark}
  With the notation above we have
  \[
  \sum_{n=1}^{n_{\rm max}}\cos\big(n (\sigma -{\sf h}\cdot V_h)\big) =
  \frac{1}{2}\Bigl(D_{n_{\rm max}}(\sigma-{\sf h}\cdot V_h) - 1\Bigr),
  \]
  that for $n_{\rm max}\to\infty$ converges in the sense of distributions
  to the Dirac delta $\delta_{\sigma_c}$ centered in $\sigma_c:={\sf
    h}\cdot V_h$.
  \end{remark}

\begin{remark}
The component
  $\frac{\partial \mathscr{H}}{\partial \sigma}$ is 
locally Lipschitz-continuous.   
\end{remark}


\subsection{Proof of Theorem \ref{teosalto}}

We shall prove the result only for the maps (\ref{dercos}), the proof
for (\ref{dersin}) being similar.  Since we assume that Jupiter cannot
collide with the Sun, the term ${\bf r}_5$ will never vanish, so that
we study only the derivatives
\[ 
\frac{\partial }{\partial y_i}\int_{\T^2}\frac{1}{d(\E,V)}\cos(n{\sf
  h}\cdot V)dV
\]
for a fixed value of $n\in\mathbb{N}$.
We shall refer to some estimates and results proved in \cite{gt13}.
For the reader's convenience we collect them in Appendix \ref{appendix}. Moreover, we shall denote by $C_k$, $k=1,\dots,12$, some
positive constants independent on $\E$.

Let $\E_c$ be a non-degenerate crossing configuration.  Let us choose
two neighborhoods $\W$ of $\E_c$ and ${\cal U}$ of $(\E_c,V_h(\E_c))$,
as in Lemma~\ref{inequalities} in the Appendix. To investigate the crossing singularity we can restrict the integral above to the set
\[
\D= \{V\in\T^2:\: (V-V_h)\cdot\A_h (V-V_h)\leq r^2 \}
\]
for some $r>0$.
We first note that
\[
\begin{split}
\frac{\partial }{\partial y_i} \int_{{\cal D}}\frac{1}{d(\E,V)}\cos(n{\sf h}\cdot V)dV  &=
  \frac{\partial}{\partial
    y_i}\int_{{\cal D}}
  \Bigl(\frac{1}{d}-\frac{1}{\delta_h}\Bigr)
  \cos(n{\sf h}\cdot V)dV\\
  &+
  \frac{\partial}{\partial y_i}\biggl(\int_{{\cal D}}\frac{\cos(n{\sf h}\cdot V) - \cos(n{\sf h}\cdot V_h)}{\delta_h}dV\biggr)\\
  &+ \frac{\partial}{\partial y_i}\bigl(\cos(n{\sf h}\cdot V_h)\bigr)
\int_{{\cal D}}\frac{1}{\delta_h}dV\\
  &+  \cos(n{\sf h}\cdot V_h)\frac{\partial}{\partial
    y_i}\int_{{\cal D}}\frac{1}{\delta_h}dV,
    \end{split}
\]
and prove that the first three addenda have a continuous extension to
$\W$.  
From the estimate (\ref{axii}) the map 
\[
\W\setminus\Sigma\ni\E \mapsto \frac{\partial}{\partial y_i}
\int_{{\cal D}}\left( \frac{1}{d(\E ,V )}- \frac{1}{\delta_h(\E ,V
  )}\right) \cos(n{\sf h}\cdot V )dV
\]
admits a continuous extension to $\W$. We now prove that also the map
\begin{equation}\label{daprov2}
\W\setminus\Sigma\ni\E \mapsto \frac{\partial}{\partial y_i}
\int_{\T^2}\frac{\cos(n{\sf h}\cdot V )-\cos(n{\sf h}\cdot
  V_h)}{\delta_h(\E ,V )} dV
\end{equation}
admits a continuous extension to $\W$.
Indeed we note that
\begin{equation}
  \begin{split}
\frac{\partial}{\partial  y_i}\frac{\cos(n{\sf h}\cdot V )-\cos(n{\sf h}\cdot V_h)}{\delta_h(\E ,V )} &=  
\frac{\sin(n{\sf h}\cdot V_h)n{\sf h}\cdot \frac{\partial V_h}{\partial  y_i}}{\delta_h(\E ,V )} \\   
&-[\cos(n{\sf h}\cdot V )-\cos(n{\sf h}\cdot V_h)]\frac{\partial}{\partial  y_i}\frac{1}{\delta_h(\E ,V )}.
\end{split}
\label{derdiffcos}
\end{equation}
By (\ref{aiii}), (\ref{axiii}) the first addendum in the r.h.s. of
\eqref{derdiffcos} is summable.
For the second, by (\ref{av}) we get
\begin{equation*}
  \left|\frac{\partial}{\partial  y_i}\frac{1}{\delta_h}\right|
  =\left|\frac{1}{2\delta_h^3}\frac{\partial\delta_h^2}{\partial y_i}\right| 
\leq \frac{C_{1}}{d_h^2 + |V-V_h|^2}.
\end{equation*}
From the estimate
\begin{equation*}
  |\cos(n{\sf h}\cdot V )-\cos(n{\sf h}\cdot V_h)| \leq C_{2}|V -V_h|
\end{equation*}
we can conclude using (\ref{avi}).

The existence of a continuous extension to $\W$ of the maps
\begin{align*}
  \W\setminus\Sigma\ni\E \mapsto&\frac{\partial}{\partial  y_i} \bigl(\cos(n{\sf h}\cdot V_h)\bigr)\int_{{\cal D}}\frac{1}{\delta_h(\E ,V )} dV =\nonumber\\
& -\sin(n{\sf h}\cdot V_h)n{\sf h}\cdot \frac{\partial V_h}{\partial  y_i}\int_{\T^2}\frac{1}{\delta_h(\E ,V )} dV  +   \cos(n{\sf h}\cdot V_h)\frac{\partial}{\partial  y_i}\int_{\T^2}\frac{1}{\delta_h(\E ,V )} dV .
\label{dercos1sud}
\end{align*}
comes from (\ref{aiii}). 

The last term cannot be extended with continuity at crossings. Using
Lemma~\ref{lem3} we define the two maps
\[
\begin{split}
 \W\ni\E\mapsto \biggl(\frac{\partial}{\partial
    y_i}\int_{{\cal D}}\frac{1}{\delta_h}dV\biggr)^\pm_h &=
  \frac{\partial}{\partial y_i}\biggl(\frac{2\pi}{\sqrt{\det{\cal A}_h}}\biggr)
  \Bigl(\sqrt{d_h^2+r^2}\mp\tilde{d}_h\Bigr)\\
  &+
  \frac{2\pi}{\sqrt{\det{\cal A}_h}}
  \biggl(\frac{\tilde{d}_h}{\sqrt{d_h^2+r^2}}
  \frac{\partial\tilde{d}_h}{\partial y_i} \mp \frac{\partial\tilde{d}_h}{\partial y_i}\biggr)
    \end{split}
\]
that are continuous extensions to $\W$ of the restrictions of
$\frac{\partial}{\partial y_i}\int_{{\cal D}}\frac{1}{\delta_h}dV$ to
$\W^\pm$ respectively.
Then we set
\[
\begin{split}
\W\ni\E\mapsto\biggl(\frac{\partial I_5^{c,n}}{\partial y_i}\biggr)^\pm_h &=
  \frac{\partial}{\partial
    y_i}\int_{{\cal D}}
  \Bigl(\frac{1}{d}-\frac{1}{\delta_h}\Bigr)
  \cos(n{\sf h}\cdot V)dV\\
  &+
  \frac{\partial}{\partial y_i}\biggl(\int_{{\cal D}}\frac{\cos(n{\sf h}\cdot V) - \cos(n{\sf h}\cdot V_h)}{\delta_h}dV\biggr)\\
  &+ \frac{\partial}{\partial y_i}\bigl(\cos(n{\sf h}\cdot V_h)\bigr)
\int_{{\cal D}}\frac{1}{\delta_h}dV\\
  &+  \cos(n{\sf h}\cdot V_h)\biggl(\frac{\partial}{\partial
    y_i}\int_{{\cal D}}\frac{1}{\delta_h}dV\biggr)^\pm_h.
    \end{split}
\]
To conclude the proof we just need to prove that these maps are
Lipschitz-continuous.  We establish the result by proving that the
function
\[
F(\E)= \int_{\D}\cos(n{\sf h}\cdot V)\frac{\partial }{\partial y_i\partial y_j}\frac{1}{d(\E,V)}dV
\]
is uniformly bounded in $\W\setminus\Sigma$.
Let us consider the Taylor expansion 
\begin{equation*}
\cos(n{\sf h}\cdot V) = \cos(n{\sf h}\cdot V_h) -n\sin(n{\sf h}\cdot V_h) {\sf h}\cdot(V-V_h) +{\cal R}^{(h)}_2,
\end{equation*}
where
\[
  {\cal R}^{(h)}_2 = {\cal R}^{(h)}_2(\E,V)
  \]
is the remainder in integral form, so that in ${\cal U}$ we have
\begin{equation}
  |{\cal R}^{(h)}_2|\leq C|V-V_h|^2
  \label{R2est}
  \end{equation}
for some $C>0$. Using the approximated distance $\delta_h$ defined in
(\ref{approxdist}) we can write $F(\E)$ as sum of four terms:
\begin{equation*}
  F= F_1 + F_2 + F_3 + F_4,
\end{equation*}
where
\begin{align*}
  &F_1 = \cos(n{\sf h}\cdot V_h)\int_{\D}\frac{\partial^2 }{\partial y_i\partial y_j}\frac{1}{d(\E,V)}dV 
,\\
&F_2= - n\sin(n{\sf h}\cdot V_h)\int_{\D}{\sf h}\cdot(V-V_h)\frac{\partial^2 }{\partial y_i\partial y_j}\left(\frac{1}{d(\E,V)}-\frac{1}{\delta_h(\E)}\right)dV
,\\
& F_3 = -n\sin(n{\sf h}\cdot V_h)\int_{\D}{\sf h}\cdot(V-V_h)\frac{\partial^2 }{\partial y_i\partial y_j}\frac{1}{\delta_h(\E)}dV 
,\\
       &F_4 = \int_{\D}{\cal R}^{(h)}_2\frac{\partial^2 }{\partial y_i\partial y_j}\frac{1}{d(\E,V)}dV.
\end{align*}
We prove that each term $F_i$ is bounded by a constant independent on $\E$.
%
%
The boundedness of $F_1$ comes trivially from (\ref{aiv}).
%
From the relation
\[ 
\frac{\partial }{\partial y_i\partial
  y_j}\frac{1}{d}=\frac{3}{4}\frac{1}{d^5}\frac{\partial d^2}{\partial
  y_i}\frac{\partial d^2}{\partial
  y_j}-\frac{1}{2}\frac{1}{d^3}\frac{\partial^2 d^2}{\partial
  y_i\partial y_j}
\] 
and the estimates (\ref{aii}),(\ref{av}),(\ref{avii}) we obtain
\[ 
\left|\frac{\partial }{\partial y_i\partial y_j}\frac{1}{d}\right| 
\leq C_3\biggl[\frac{1}{d^5}(d_h+|V-V_h|)^2 +
\frac{1}{d^3}\biggr] \leq \frac{C_4}{(d_h^2+|V-V_h|^2)^{3/2}}.
\]
Then (\ref{R2est}) and (\ref{avi}) yield the boundedness of $F_4$:
\[
\left| \int_{\D}{\cal R}_2^{(h)}\frac{\partial }{\partial y_i\partial y_j}\frac{1}{d(\E,V)}dV    \right|
%
\leq C_5\int_{\cal D}\frac{dV}{d_h +|V-V_h|} \leq C_6. 
\]
To show the boundedness of $F_2$ we just need to prove that
\begin{equation}\label{daprov}
\left|\frac{\partial^2 }{\partial y_i\partial y_j}\left(\frac{1}{d}-\frac{1}{\delta_h}\right)\right|\leq\frac{C_{7}}{d_h^2+|V-V_h|^2},
\end{equation}
so that
\[
\left|\int_{\D}{\sf h}\cdot(V-V_h)\frac{\partial^2 }{\partial
  y_i\partial y_j}\left(\frac{1}{d}-\frac{1}{\delta_h}\right)dV\right|
\leq C_{8} \int_{\D} \frac{dV}{d_h+|V-V_h|}\leq C_9.
\]
Using $d^2 = \delta_h^2+{\cal R}_3^{(h)}$ we get
\begin{align*}
\frac{\partial^2 }{\partial y_i\partial y_j}\left(\frac{1}{d}-\frac{1}{\delta_h}\right) &=  
\frac{3}{4}\left( \frac{1}{d^5}\frac{\partial d^2}{\partial y_i}
-\frac{1}{\delta_h^5}\frac{\partial \delta_h^2}{\partial y_i} \right)\frac{\partial \delta_h^2}{\partial y_j}
 + \frac{1}{2}\left( \frac{1}{d^3}-\frac{1}{\delta_h^3} \right)\frac{\partial^2 \delta_h^2}{\partial y_i\partial y_j} \\
& +\frac{3}{4}\frac{1}{d^5}\frac{\partial d^2}{\partial y_i}\frac{\partial {\cal R}_3^{(h)}}{\partial y_j} -\frac{1}{2}\frac{1}{d^3}\frac{\partial^2 {\cal R}_3^{(h)}}{\partial y_i\partial y_j}.
\end{align*}
We prove that each of the four terms in the previous sum satisfies an
estimate like (\ref{daprov}). For the second term we use estimates
(\ref{avii}),(\ref{aviii}), for the third 
(\ref{av}),(\ref{aix}), and for the last (\ref{ax}). To estimate
the first term we note that
\[ 
\left( \frac{1}{d^5}\frac{\partial d^2}{\partial y_i}
-\frac{1}{\delta_h^5}\frac{\partial \delta_h^2}{\partial y_i} \right)\frac{\partial \delta_h^2}{\partial y_j} = 
\left( \frac{1}{d^5} -\frac{1}{\delta_h^5}\right)\frac{\partial \delta_h^2}{\partial y_i}\frac{\partial \delta_h^2}{\partial y_j}+ \frac{1}{d^5}\frac{\partial {\cal R}_3^{(h)}}{\partial y_i}\frac{\partial \delta_h^2}{\partial y_j}
\]   
and use
\[
\left| \frac{1}{d^5} -\frac{1}{\delta_h^5}\right|\leq\left| \frac{1}{d} -\frac{1}{\delta_h}\right|\left| \frac{1}{d^4} +\frac{1}{d^3\delta_h}+\frac{1}{d^2\delta_h^2}+\frac{1}{d\delta_h^3}+\frac{1}{\delta_h^4}\right|. 
\]
We can conclude using (\ref{aii}),(\ref{av}),(\ref{aix}),(\ref{axi}).

Now we show the boundedness of $F_3$. We write
\begin{align}
 \int_{\D}{\sf h}\cdot(V-V_h)\frac{\partial^2 }{\partial y_i\partial
   y_j}\frac{1}{\delta_h}&= \frac{3}{4}\int_{\D}{\sf
   h}\cdot(V-V_h)\frac{1}{\delta_h^5}\frac{\partial
   \delta_h^2}{\partial y_i}\frac{\partial \delta_h^2}{\partial
   y_j}dV \nonumber \\
  &-\frac{1}{2}\int_{\D}{\sf
   h}\cdot(V-V_h)\frac{1}{\delta_h^3}\frac{\partial^2
   \delta_h^2}{\partial y_i\partial y_j}dV
\label{mammamia}
\end{align}
and study the two integrals in the r.h.s. separately. To estimate the first
we use (\ref{approxdist}) and get 
\[
\frac{\partial \delta_h^2}{\partial y_j} = \frac{\partial d_h^2}{\partial y_j} -2\frac{\partial V_h}{\partial y_j}\cdot\A_h(V-V_h)+(V-V_h)\cdot\frac{\partial \A_h}{\partial y_j}(V-V_h),
\]
so that
\begin{align*}
\frac{\partial \delta_h^2}{\partial y_i}\frac{\partial \delta_h^2}{\partial y_j} & = 
 \frac{\partial d_h^2}{\partial y_i} \frac{\partial d_h^2}{\partial y_j} -2\left(\frac{\partial d_h^2}{\partial y_i} \frac{\partial V_h}{\partial y_j}+ 
\frac{\partial d_h^2}{\partial y_j} \frac{\partial V_h}{\partial y_i} \right)\cdot\A_h(V-V_h)        \\
& +  \frac{\partial d_h^2}{\partial y_i}(V-V_h)\cdot\frac{\partial \A_h}{\partial y_j}(V-V_h)+\frac{\partial d_h^2}{\partial y_j}(V-V_h)\cdot\frac{\partial \A_h}{\partial y_i}(V-V_h)   \\
&+ 4\left[\frac{\partial V_h}{\partial y_i}\cdot\A_h(V-V_h)\right]\left[\frac{\partial V_h}{\partial y_j}\cdot\A_h(V-V_h)\right] \\
& -2\left[\frac{\partial V_h}{\partial y_i}\cdot\A_h(V-V_h)\right]\left[(V-V_h)\cdot\frac{\partial \A_h}{\partial y_j}(V-V_h)\right] \\
& -2\left[\frac{\partial V_h}{\partial y_j}\cdot\A_h(V-V_h)\right]\left[(V-V_h)\cdot\frac{\partial \A_h}{\partial y_i}(V-V_h)\right] \\
& +  \left[(V-V_h)\cdot\frac{\partial \A_h}{\partial y_i}(V-V_h)\right]\left[(V-V_h)\cdot\frac{\partial \A_h}{\partial y_j}(V-V_h)\right].
\end{align*}
Then we use the change of variables $\xi =\A_h^{1/2}(V-V_h)$ and polar
coordinates $(\rho,\theta)$ defined by $\xi=\rho(\cos\theta, \sin\theta)$.
We distinguish between terms with even and odd degree in $(V-V_h)$.
First we consider the ones with even degree. The term of degree $2$ is
estimated as follows
\begin{align*}
&\left|\int_{\D}{\sf h}\cdot(V-V_h)\frac{1}{\delta_h^5}\left(\frac{\partial
    d_h^2}{\partial y_i} \frac{\partial V_h}{\partial
    y_j}+\frac{\partial d_h^2}{\partial y_j} \frac{\partial
    V_h}{\partial y_i} \right)\cdot\A_h(V-V_h)dV\right| = \\ &
  \biggl|\int_{\D} 2\tilde{d}_h\left(\frac{\partial \tilde{d}_h}{\partial y_i} \frac{\partial
    V_h}{\partial y_j}+\frac{\partial \tilde{d}_h}{\partial y_j}
  \frac{\partial V_h}{\partial y_i}
  \right)\cdot \A_h(V-V_h){\sf
    h}\cdot(V-V_h)\frac{1}{\delta_h^5}dV\biggr| =\\ &
  2\frac{d_h}{\sqrt{\det\A_h }}
  \int_0^r \frac{\rho^3}{(d_h^2+\rho^2)^{5/2}}d\rho
  \left|\sum_{|\gamma|=2}b_\gamma\int_0^{2\pi}
  (\cos\theta)^{\gamma_1}(\sin\theta)^{\gamma_2}d\theta \right|\leq \\ &
  2\frac{d_h}{\sqrt{\det\A_h}}\frac{C_{10}}{d_h} \leq C_{11},
\end{align*}
while for the term of degree $4$ we note that
\begin{align*}
& \Biggl|\int_{\D}\frac{1}{\delta_h^5}{\sf h}\cdot(V-V_h)\left[\frac{\partial
      V_h}{\partial
      y_j}\cdot\A_h(V-V_h)\right]\left[(V-V_h)\cdot\frac{\partial
      \A_h}{\partial y_i}(V-V_h)\right]dV \Biggr|=\\ &
  \frac{1}{\sqrt{\det\A_h }}
\int_0^r\frac{\rho^5}{(d_h^2+\rho^2)^{5/2}}d\rho
  \left|\sum_{|\gamma|=4}c_\gamma\int_0^{2\pi}
  (\cos\theta)^{\gamma_1}(\sin\theta)^{\gamma_2}d\theta \right|
 \leq C_{12} 
\end{align*}
for some functions $b_\gamma$, $c_\gamma$, uniformly bounded in ${\cal
  W}\setminus\Sigma$, and for
$\gamma=(\gamma_1,\gamma_2)\in(\mathbb{N}\cup \{0\})^2$.
The terms with odd degree in $(V-V_h)$ vanish, as can be shown
by similar computations, using
\[
\int_0^{2\pi} (\cos\theta)^{\gamma_1}(\sin\theta)^{\gamma_2}d\theta = 0
\]
with $\gamma_1 + \gamma_2$ odd.
To estimate the second integral in (\ref{mammamia}) we proceed in a
similar way, using
\begin{align*}
\frac{\partial^2 \delta_h^2}{\partial y_i\partial y_j}& =
\frac{\partial^2 d_h^2}{\partial y_i\partial y_j} -2\frac{\partial^2
  V_h}{\partial y_i\partial y_j}\cdot\A_h(V-V_h) -2 \frac{\partial
  V_h}{\partial y_j}\cdot\frac{\partial \A_h}{\partial y_i}(V-V_h) \\
& -2 \frac{\partial V_h}{\partial y_i}\cdot\frac{\partial
  \A_h}{\partial y_j}(V-V_h) +\left[(V-V_h)\cdot\frac{\partial^2
    \A_h}{\partial y_i\partial y_j }(V-V_h)\right].
\end{align*}

\begin{remark}
  If $\mathcal{E}_c$ is an orbit configuration with two crossings,
  assuming that $d_h(\mathcal{E}_c) = 0$ for $h=1,2$, we can extract
  the singularity by considering the approximated distances $\delta_1,
  \delta_2$ and considering $1/d$ as sum of the three terms $(1/d -
  1/\delta_1 - 1/\delta_2)$, $1/\delta_1$, $1/\delta_2$.
\end{remark}

\section{Generalized solutions and evolution of the orbit distance}
\label{sec:gen}

Following \cite[Sections 5-6]{gt13} we can construct generalized
solutions by patching classical solutions defined in the domain $\W^+$
with classical solutions defined on $\W^-$ and vice-versa.
Let $(E(t),\sigma(t))$, with $E(t)=(S(t),G(t),Z(t),g(t),z(t))$,
represent the evolution of the asteroid according to
(\ref{hamres}). In a similar way we denote by $E'(t)$ a known function
of time representing the evolution of the trajectory of the planet.
Setting $\E(t)=( E(t),E'(t) )$ we let $T(\Y)$ be the set of times
$t_c$ such that $d_{min}(\E(t_c))=0$ and suppose that it has no
accumulation points.

We say that $\Y(t)$ is a {\it generalized solution} of (\ref{hamres})
if it is a classical solution for $t\notin T(\Y)$ and for each $t_c\in
T(\Y)$ there exist finite values of
\[
\lim_{t\to t_c^+}\dot\Y(t), \qquad \lim_{t\to t_c^-}\dot\Y(t).
\]  

In order to construct a generalized solution we consider a solution
$\Y(t)$ of the Cauchy problem given by (\ref{hamres}) with a non
crossing initial condition $\Y(t_0)$. Suppose that it is defined on a
maximal interval $J$ such that $\sup J =t_c\in T(\Y)$ and that
$\Y(t)\in\W^+$ as $t\to t_c$ . Suppose that the crossing is occurring
with a planet different from Jupiter (resp. Jupiter itself).  Applying
Theorem \ref{gttheorem}-(a) (resp. Theorems \ref{gttheorem}-(a) and
\ref{teosalto}-(a)) we have that there exists
\[
\lim_{t\to t_c^-}\dot\Y(t) = \dot{\Y}_c 
\]
and the solution can be extended beyond $t_c$ considering the Cauchy
problem
\begin{equation*}
\dot{\Y} = \mathbb{J}_3(\nabla_\Y\mathscr{H})^+, \quad \Y(\tau)=\Y_\tau
\end{equation*}
for some $\tau\to t_c$, so that we call $\Y(t_c)=\Y_c$. 
Using again Theorem \ref{gttheorem}-(a)
(resp. Theorems \ref{gttheorem}-(a) and \ref{teosalto}-(a)), we can
extend the solution beyond the singularity considering the new Cauchy
problem
\begin{equation*}
\dot{\Y} = \mathbb{J}_3(\nabla_\Y\mathscr{H})^-, \quad \Y(t_c)=\Y_c.
\end{equation*}  
whose solution fulfills, from Corollary \ref{cor1} (resp. Corollary
\ref{cor2})
\[
\lim_{t\to t_c^-}\dot\Y(t) = \dot{\Y}_c - \mathrm{Diff}_h\left(\nabla_\Y\mathscr{H}\right)(\E(t_c),V).
\]
Note that the evolution of the orbital elements according to a
generalized solution is continuous but not differentiable in a
neighborhood of a crossing singularity. More precisely, the evolution
of the elements $(G,Z,\sigma,g,z)$ is only Lipschitz-continuous while
the evolution of $S$ is $C^1$, since
$\frac{\partial\mathscr{H}}{\partial\sigma}$ is continuous also at orbit crossings.
\\ Once a generalized solution $\Y(t)=(E(t),\sigma(t))$ is defined, we
can consider the evolution of the distance $\tilde{d}_h(\E(t))$.  Let
us define
$$
\bar d_h(t) = \tilde{d}_h(\E(t))
$$ 
and suppose that it is defined in an interval containing a crossing
time $t_c$ corresponding to a non-degenerate crossing
configuration. We have the following
\begin{prop}
Let $\Y(t)$ be a generalized solution of (\ref{hamres}) and $\E(t)$ be
defined as above. Suppose that $t_c$ is a crossing time such that
$\E_c=\E(t_c)$ is a non-degenerate crossing configuration. Then there
exists an open interval $I\ni t_c$ such that $\bar d_h\in
C^1(I,\mathbb{R})$.
\end{prop}

\begin{proof}
We choose the interval $I$ such that $\E(I)\in\W$ with $\W$ defined in
Theorem \ref{gttheorem} (resp. \ref{teosalto}) and suppose that
$\E(t)\in\W^+$ for $t<t_c$ and $\E(t)\in\W^-$ for $t>t_c$. We can
compute, for $t\neq t_c$,
\begin{align*}
 \dot{\bar{d_h}}(t) & = \nabla_\E\tilde{d}_h(\E(t))\cdot\dot\E(t) = \nabla_E\tilde{d}_h(\E(t))\cdot\dot E(t) + \nabla_{E'}\tilde{d}_h(\E(t))\cdot\dot E'(t) \\
 &=  \nabla_E \tilde{d}_h(\E(t)) \cdot\Bigl( -\frac{\partial \mathscr{H}}{\partial \sigma},-\frac{\partial \mathscr{H}}{\partial g},-\frac{\partial \mathscr{H}}{\partial z},\frac{\partial \mathscr{H}}{\partial G},\frac{\partial \mathscr{H}}{\partial Z} \Bigr)^T + \nabla_{E'}\tilde{d}_h(\E(t))\cdot\dot E'(t).
\end{align*}
The second addendum is continuous while for the first we need to
distinguish between crossing a planet different from Jupiter (the
resonant planet) and crossing Jupiter itself. In the first case, we
apply Corollary \ref{cor1} and obtain
\begin{align*}
\lim_{t\to t_c^+} \dot{\bar{d_h}}(t) - \lim_{t\to t_c^-} \dot{\bar{d_h}}(t)   &=  
\biggl[\nabla_E\tilde d_h \cdot \mathrm{Diff}_h\Bigl.\Bigl(-\frac{\partial \mathscr{H}}{\partial \sigma},-\frac{\partial \mathscr{H}}{\partial g},-\frac{\partial \mathscr{H}}{\partial z},\frac{\partial \mathscr{H}}{\partial G},\frac{\partial \mathscr{H}}{\partial Z} \Bigr)^T\biggr]_{t=t_c} \\
&=\biggl[\nabla_E\tilde d_h \cdot \mathrm{Diff}_h\Bigl.\Bigl( 0,-\frac{\partial \mathscr{H}}{\partial g},-\frac{\partial \mathscr{H}}{\partial z},\frac{\partial \mathscr{H}}{\partial G},\frac{\partial \mathscr{H}}{\partial Z} \Bigr)^T\biggr]_{t=t_c}  \\
&  =\biggl[ -\frac{2\mu_5 \k^2}{\pi\sqrt{\det(\mathcal{A}_h)}}\{\tilde{d_h},\tilde d_h  \}\biggr]_{t=t_c} = 0, 
\end{align*}
where $\{,\}$ are the Poisson brackets.

In the second case, we apply Corollary \ref{cor2} and get
\begin{align*}
\lim_{t\to t_c^+} \dot{\bar{d_h}}(t) - \lim_{t\to t_c^-} \dot{\bar{d_h}}(t)   &=  
\biggl[\nabla_E\tilde d_h \cdot \mathrm{Diff}_h\left( -\frac{\partial \mathscr{H}}{\partial \sigma},-\frac{\partial \mathscr{H}}{\partial g},-\frac{\partial \mathscr{H}}{\partial z},\frac{\partial \mathscr{H}}{\partial G},\frac{\partial \mathscr{H}}{\partial Z} \right)^T\biggr]_{t=t_c} \\
&=\biggl[\nabla_E\tilde d_h \cdot \mathrm{Diff}_h\left( 0,-\frac{\partial \mathscr{H}}{\partial g},-\frac{\partial \mathscr{H}}{\partial z},\frac{\partial \mathscr{H}}{\partial G},\frac{\partial \mathscr{H}}{\partial Z} \right)^T\biggr]_{t=t_c}  \\
& = \biggl[
  -\frac{2\mu_5 \k^2\left[ \sum_{n=1}^{n_{\rm max}}\cos\big(n(\sigma-{\sf h}\cdot V_h)\big)+ \frac{1}{2}\right]}{\pi\sqrt{\det(\mathcal{A}_h)}}\{\tilde{d_h},\tilde d_h  \}\biggr]_{t=t_c} = 0 .
\end{align*}

\end{proof}

\section{Dynamical protection from collisions}
\label{sec:dynpr}

In case of crossings with the resonant planet, the resonance protects
the asteroid from close encounters with that planet (see
\cite{milbac98}). This protection mechanisms is usually derived by a
perturbative approach different from ours. Here we describe how this
mechanism can be recovered from the normal form \eqref{hfin} in the
limit for $n_{\rm max}\to\infty$.

Let us consider, for simplicity, a restricted 3-body problem
Sun-planet-asteroid, where the asteroid is in a mean motion resonance
with the planet, given by
\[
{\sf h} = (h,h')\in\Z^2,
\] 
and their trajectories cross each other during the evolution.
In the following we take a Hamiltonian containing only the direct
part of the perturbation, the indirect part being regular. Therefore we set
\[
{\cal H}= \frac{1}{d},
\]
where $d$ is the distance between the asteroid and the planet.
We consider the following procedures:

{\bf (I)} Through a unimodular transformation $\Psi$ of the fast
variables $V=(\ell,\ell')$ we pass to new variables $(\sigma,\tau)$,
with
\[
\sigma={\sf h}\cdot V,
\]
whose evolution occurs on different time scales: $\sigma$ has a {\em
  long-term} evolution, $\tau$ has a {\em fast} evolution.
More precisely we have
\begin{equation}
V \stackrel{\Psi}{\to}W = \mathcal{U} V
\label{trasfcan}
\end{equation}
where $W=(\sigma,\tau)^T$ and $\mathcal{U}$ is a constant unimodular
matrix whose first raw is $(h, h')$.
The transformation $\Psi$ can be extended to a canonical
transformation (here denoted again by $\Psi$) by defining the
corresponding actions as $(S,T)=\mathcal{U}^{-T}(L,L')$ and leaving
the other variables unchanged.
Then, we average over the fast variable $\tau$ and get the Hamiltonian
\begin{equation}
\ov{\cal K}(\sigma,S,T;X) = \frac{1}{2\pi}\int_0^{2\pi}{\cal
  H}\circ\Psi^{-1}(\sigma,\tau,S,T;X)d\tau.
\label{kappabar}
\end{equation}
Here $X$ is the vector of the other variables, evolving on a secular
time scale. This procedure is used e.g. in \cite{milbac98}.

\medbreak {\bf (II)} As in Section \ref{sec:av_eq}, we consider the
resonant normal form obtained by eliminating all the non resonant
harmonics from the Fourier series of the Hamiltonian.  For each
integer $N$ we take the partial Fourier sums
\[
\mathscr{H}_N(V,L,L';X) = \sum_{\stackrel{|{\sf k}|\leq N}{{\sf
      k}\in{\cal R}}}\hat{\cal H}_{{\sf k}}(L,L';X)e^{i{\sf
    k}\cdot V},
\]
where 
\[
  {\cal R}=\{{\sf k}=(k,k')\in\Z^2: \exists
  n\in\Z \mbox{ with } {\sf k} = n{\sf h}\}
\]
and 
\[
\hat{\cal H}_{{\sf k}}(L,L';X) = \frac{1}{(2\pi)^2}\int_{\mathbb{T}^2}{\cal
  H}({\cal V},L,L';X)e^{-i{\sf k}\cdot {\cal V}}d{\cal V},
\]
in which we denote by ${\cal V}$ the vector $(\ell,\ell')$ when the
latter are integration variables.
We formally define
\[
\mathscr{H}_{\infty}(V,L,L';X) =
\lim_{N\to\infty}\mathscr{H}_N(V,L,L';X). 
\]
Note that
\[
\mathscr{H}_N(V,L,L';X) = \frac{1}{(2\pi)^2}\int_{\T^2}
D_N({\sf h}\cdot {\cal V} - {\sf h}\cdot V){\cal H}({\cal V},L,L';X) d{\cal V},
\]
where $D_N(x)$
is the Dirichlet kernel.
We introduce the functions
\[
\begin{split}
\mathscr{K}_N(\sigma,S,T;X) &=
\mathscr{H}_N\circ\Psi^{-1}(\sigma,\tau,S,T;X),\cr
\mathscr{K}_{\infty}(\sigma,S,T;X) &=
\mathscr{H}_{\infty}\circ\Psi^{-1}(\sigma,\tau,S,T;X).
\end{split}
\]
Indeed both $\mathscr{K}_N$ and $\mathscr{K}_{\infty}$ do not depend
on $\tau$.  The Hamiltonian $\mathscr{K}_N$ corresponds to the
resonant normal form in (\ref{hfin}). However, here we used a
unimodular matrix ${\cal U}$ in the canonical transformation.

Moreover, we observe that the Hamiltonian $\ov{\cal K}$ defined in
\eqref{kappabar} can be written as a pointwise limit for $N\to\infty$
of the partial Fourier sums
\[
\ov{\cal K}_N(\sigma,S,T;X) =
\frac{1}{2\pi}\int_0^{2\pi}D_N(\tilde{\sigma}-\sigma)\ov{\cal
  K}(\tilde{\sigma},S,T;X)d\tilde{\sigma}.
\]
Let $\sigma_c={\sf h}\cdot V_h$. If $d_h=0$ then $\sigma_c$ is the
value of $\sigma$ allowing a collision, occurring for $V=V_h$.
Assume that ${\cal E}_c$ is a non-degenerate crossing configuration,
i.e. $d_h=0$ and ${\cal A}_h$ is positive definite.
We use $Y = Y({\cal E})$ to denote the variables different from
$\sigma$ and we set $Y_c=Y({\cal E}_c)$.

\begin{prop}
The following properties
hold.
\begin{enumerate}

\item
  If ${\cal E}\neq{\cal E}_c$, then for each $\sigma$ we have
  \begin{itemize}
  \item[i)]    $\mathscr{K}_N(\sigma;Y) = \ov{\cal K}_N(\sigma;Y),\quad \forall N$    
  \item[ii)] $\mathscr{K}_\infty(\sigma;Y) = \ov{\cal K}(\sigma;Y) $.
\end{itemize}
   Moreover, these functions are differentiable with
  continuity with respect to $Y$.

 \item For ${\cal E}={\cal E}_c$ we have
\begin{itemize}
\item[i)]   $\mathscr{K}_N(\sigma;Y_c) = \ov{\cal K}_N(\sigma;Y_c),\qquad \forall N, \ \forall \sigma$
\item[ii)]    $\mathscr{K}_\infty(\sigma;Y_c) = \ov{\cal K}(\sigma;Y_c), \quad \qquad \forall \sigma\neq\sigma_c$
\item[iii)]     $\lim_{\sigma\to\sigma_c}\mathscr{K}_{\infty}(\sigma;Y_c) =
     \lim_{\sigma\to\sigma_c}\ov{\cal K}(\sigma;Y_c) = +\infty$.
\end{itemize}
      

\item If ${\cal E}={\cal E}_c$ and $\sigma\neq\sigma_c$ then, denoting
  by $y_j$ a generic component of $Y$,
  \begin{itemize}
    \item[i)] the derivatives $\displaystyle\frac{\partial\mathscr{K}_\infty}{\partial y_j}(\sigma;Y_c) = \displaystyle\frac{\partial\ov{\cal
        K}}{\partial y_j}(\sigma;Y_c)$ exist and are continuous;
   \item[ii)] the derivatives 
    $\displaystyle\frac{\partial\mathscr{K}_N}{\partial y_j}(\sigma;Y_c) = \displaystyle\frac{\partial\ov{\cal K}_N}{\partial y_j}(\sigma;Y_c)$
     generically do not exist.
     
\end{itemize}

\item For each $N$ and for each value of $\sigma$ there exist the limits
  \[
  \lim_{{\cal E}\to{\cal E}_c^\pm}\frac{\partial\mathscr{K}_N}{\partial y_j}(\sigma;Y) \left(= \lim_{{\cal E}\to{\cal E}_c^\pm}\frac{\partial\ov{\cal K}_N}{\partial y_j}(\sigma;Y)\right)
  \]
  from both sides of the crossing configuration set $\Sigma$. These
  limits are generically different and their difference 
  converges in the sense of distributions, for $N\to\infty$, to the
  Dirac delta relative to $\sigma_c$, multiplied by the factor
  \[
  \frac{- 2\mu'{\rm k}^2}{\sqrt{\det({\cal A}_h)}}\frac{\partial
    \tilde{d}_h}{\partial y_i}({\cal E}_c).
  \]
\end{enumerate}
\label{properties}
\end{prop}

\begin{remark}
  If ${\cal E}={\cal E}_c$, procedure {\bf (I)} gives a well defined
  vector field, provided that $\sigma\neq\sigma_c$. On the other hand,
  with procedure {\bf (II)} it does not make sense to consider
  \[
  \lim_{N\to\infty} \frac{\partial
    \mathscr{K}_N}{\partial y_j}(\sigma;Y_c).
  \]
However, for each $N$ we can extend the vector field of
$\mathscr{K}_N$ in two different ways on $\Sigma$, and the difference
between the two extensions has a very weak behavior for $N\to\infty$:
it tends to a Dirac delta in the sense of distribution, being the
singularity of the delta just at $\sigma=\sigma_c$.
\end{remark}


\noindent{\em Proof of Proposition~\ref{properties}}.
  
  1. For every $N$,
by applying the change of variables $V\to
\Psi(V)$ and Fubini-Tonelli's theorem we obtain
\begin{align}
  \begin{split}\label{K_n}
  \mathscr{K}_N(\sigma;Y)&=\mathscr{H}_N\circ\Psi^{-1}(\sigma,\tau;Y)
  =\frac{1}{(2\pi)^2}\int_{\T^2}
D_N({\sf h}\cdot {\cal V} - \sigma){\cal H}({\cal V};Y) d{\cal
  V}\\
&=\frac{1}{(2\pi)^2}\int_{\T^2} D_N(\tilde\sigma - \sigma){\cal
  H}\circ \Psi^{-1}(\tilde{\sigma},\tilde{\tau};Y)
d\tilde{\sigma}d\tilde{\tau}\\
&=\frac{1}{2\pi}\int_{\T} D_N(\tilde\sigma - \sigma)\left(
\frac{1}{2\pi}\int_{\T}{\cal H}\circ
\Psi^{-1}(\tilde{\sigma},\tilde{\tau};Y) d\tilde{\tau}
\right)d\tilde{\sigma}\\
&= \frac{1}{2\pi}\int_{\T} D_N(\tilde\sigma - \sigma)\ov{\cal
  K}(\tilde\sigma;Y)d\tilde\sigma = \ov{\cal K}_N(\sigma;Y),
\end{split}
\end{align}
that proves i). Point ii) comes from the fact
that, for ${\cal E}\neq {\cal E}_c$, $\ov{\cal K}(\sigma;Y)$ is a
smooth function of $\sigma$ and the corresponding Fourier series
converge pointwise for every $\sigma$. Hence we can pass to the limit
as $N\to\infty$ in the previous equality.

\noindent The differentiability comes from the fact that the distance
function ${\cal H} = 1/d$ is bounded for ${\cal E}\neq {\cal
  E}_c$.

\bigbreak 2.  To prove i), we can repeat the argument used in
(\ref{K_n}). Indeed, the double integral is finite also for ${\cal
  E}={\cal E}_c$ and we can apply Fubini-Tonelli's theorem.

\noindent To prove ii), we recall that the Fourier series of an $L^1$
function
converges pointwise at every point of differentiability
\cite{stein}. Therefore, for every $\sigma\neq\sigma_c$, $\ov{\cal
  K}_N(\sigma;Y_c)\to\ov{\cal K}(\sigma;Y_c)$ for $N\to\infty$. Hence,
using i) and passing to the limit for $N\to\infty$ in $\mathscr{K}_N$
we get the result.

\noindent To prove iii) we just need to prove that one of the two limits 
diverges. From Fatou's lemma
\begin{align}
\liminf_{\sigma\to\sigma_c}\ov{\cal K}(\sigma;Y_c) &\geq
\frac{1}{2\pi}\int_0^{2\pi}{\cal
  H}\circ\Psi^{-1}(\sigma_c,\tau;Y_c)d\tau
\nonumber\\
&=\frac{1}{2\pi}\int_0^{2\pi}\frac{1}{
  d\circ\Psi^{-1}(\sigma_c,\tau;Y_c)}d\tau = +\infty.
\label{diverges}
\end{align}
We can prove that the integral in \eqref{diverges}
diverges by a singularity extraction technique. Let us write
\begin{equation}
\frac{1}{d} = \Bigl(\frac{1}{d} -\frac{1}{\delta_h}\Bigr) +\frac{1}{\delta_h}.
\label{singextr}
\end{equation}
The first term in the r.h.s. of \eqref{singextr} is bounded, while the
integral of the second diverges because
\begin{equation}
\delta_h^2({\cal U}^{-1}Z + V_h;Y_c) = Z\cdot{\cal B}_h Z \geq
\frac{\det{\cal A}_h}{b_{22}}(\sigma-\sigma_c)^2,
\label{estim}
\end{equation}
where 
\[
Z = {\cal U}(V-V_h),
\] 
with ${\cal U}$ the unimodular matrix
defined in \eqref{trasfcan}, and
\[
{\cal B}_h = U^{-T}{\cal A}_h U^{-1}.
\]
The number $b_{22}$ in \eqref{estim} is defined by
\[
b_{22}=e_2\cdot{\cal B}_he_2  
\]
and is strictly positive because ${\cal B}_h$ is positive definite, being ${\cal
  E}_c$ non-degenerate (and therefore ${\cal A}_h$ positive definite).

\bigbreak 3. Estimate \eqref{estim}, decomposition \eqref{singextr},
and the theorem of differentiation under the integral sign yield the
existence and continuity of the derivatives
$\displaystyle\frac{\partial\ov{\cal K}}{\partial y_j}$, that is i).
Point ii) is a consequence of property 4.

\bigbreak
4. This follows from Theorem~\ref{teosalto} and Corollary~\ref{cor2}.


\hfill{$\square$}

\section{Numerical experiments}\label{sec:numerics}

We compare the long term evolution coming from system (\ref{hamres})
with the full evolution of equation (\ref{Nbodyeq}),
corresponding to the classical restricted $N$-body problem.

To get the evolution of the planets, we compute a planetary
ephemerides database for a time span of 2000 yrs, starting at 57600
MJD with a time step of 0.5 years. The computation is performed using
the FORTRAN program {\tt orbit9} included in the {\tt OrbFit} free
software\footnote{{\tt http://adams.dm.unipi.it/orbmaint/orbfit}}. The
planetary evolution at the desired time is obtained from this
database by linear interpolation.

Inspired by the classification in \cite{milani89} we consider two
paradigmatic cases, representing the two crossing behaviors discussed
in the previous sections. The first case is asteroid (887) Alinda,
that is considered in the gravitational field of 5 planets, from Venus
to Saturn. This asteroid is in $3:1$ mean motion resonance with
Jupiter and we will consider its crossings with the orbit of Mars. The
second case deals with the 'Toro' class: we consider a fictitious
asteroid that we call 1685a under the influence of 3 planets: the
Earth, Mars and Jupiter. This asteroid crosses the orbit of the Earth,
and is in the $5:8$ mean motion resonance with it.

We use the same algorithm as in \cite{gt13} to compute the solution of
system (\ref{hamres}). This is a Runge-Kutta-Gauss method evaluating
the vector field at intermediate points of the time step. The time
step is reduced when the trajectory of the asteroid is close to a
planet crossing, in order to get exactly the crossing condition.  By
Theorems \ref{gttheorem}-\ref{teosalto} we can find two
locally Lipschitz-continuous extensions of the vector field from both sides of
the singular set $\Sigma$. The difference between the two extended
fields is given by Corollary \ref{cor1} for asteroid 887 (Alinda) and
by Corollary \ref{cor2} for asteroid 1685a. In both cases, we compute
the intermediate values of the extended vector field just after the
crossing, and then we correct them using Corollary \ref{cor1} or
Corollary \ref{cor2}. We use these corrected values as an
approximation of the vector field at the intermediate point
of the solution, see Figure \ref{fig:RKG}. This algorithm avoids the
computation of the vector field at the singular points, which could be
affected by numerical instability.

\begin{figure}[ht!]
  \psfragscanon \psfrag{A}{$\Y_{k-1}$} \psfrag{B}{$\Y_k$}
  \psfrag{C}{$\Y_{k+1}$} \psfrag{Wp}{${\cal W}^+$}\psfrag{Wm}{${\cal
      W}^-$} \psfrag{sigma}{$\Sigma$}
  \centerline{\epsfig{figure=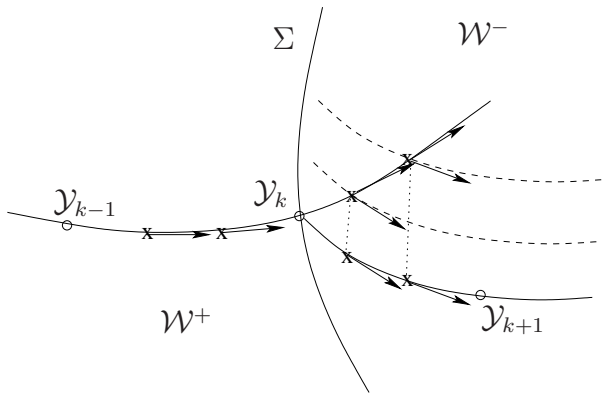,width=8cm}}
  \psfragscanoff
\caption{Runge-Kutta-Gauss method and continuation of the solution of
  (\ref{hamres}) beyond the singularity.}
\label{fig:RKG}
\end{figure}

To produce the comparison, we consider 64 possible initial conditions
for system (\ref{Nbodyeq}) corresponding to the same initial
condition of system (\ref{hamres}). For asteroid 887 (Alinda)
these are produced by shifting the mean anomalies in the following
way.  Let $\bar{\ell}_j$ and $\bar{\ell}$ be the mean anomalies of
planet $j$ and the asteroid, at the initial epoch 57600
MJD. For each planet, we consider the 64 values
$\ell^{(k)}_j=\bar{\ell}_j+k\pi/64$ with $k=0,\dots,63$. For every
$k$, we compute the initial value of the mean anomaly
$\ell^{(k)} =\bar{\ell} + l^{(k)}$ 
of the asteroid such that
\[
h^*_5 (\bar{\ell}_{5}+k\pi/64) + h^*(\bar{\ell} + l^{(k)})=
h^*_5\bar{\ell}_5+h^* \bar{\ell}.
\]  
The integration of this 64 different initial conditions is performed
with the program {\tt orbit9}. Then we consider the arithmetic mean of
the 5 Keplerian elements $a,e,I,\Omega,\omega$ and the critical angle
$\sigma=h^*_5 \ell_5+h^*\ell$ over these evolutions and compare them
with the corresponding elements coming from system (\ref{hamres}), in
which we choose $n_{\rm max} = 3$. Figure \ref{fig:887} summarizes the
results: the solid line corresponds to the solution of (\ref{hamres})
while the dashed line corresponds to the arithmetic mean of the full
numerical integrations. The shaded region represents the standard
deviation from the arithmetic mean. The correspondence between the
solutions is good. The Mars crossing singularity occurs around
$t=3786\: yr$.

\begin{figure}[h!]
\centerline{\epsfig{figure=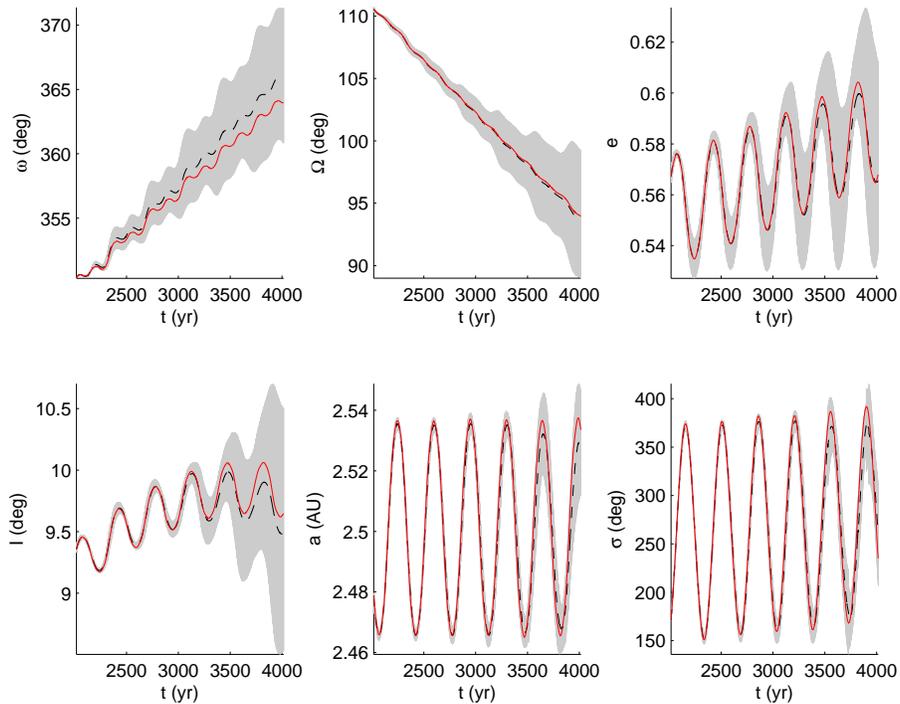,width=12cm}}
\caption{Asteroid 887 (Alinda): comparison between the long term
  evolution using $n_{\rm max} = 3$ (solid line) and the arithmetic mean of
  64 full numerical integrations (dashed line).}
\label{fig:887}
\end{figure}

For asteroid 1685a we proceed in the same way, with the Earth playing
the role of Jupiter. For the long term evolution we used $n_{\rm
  max}=3, 15$. In Figure \ref{fig:1685a} we show the results. Using
$n_{\rm max}=15$ we see that the result improves very much.  The Earth
crossing singularity occurs around $t=2281\: yr$.  In this test the
value of $\sigma_c$ at crossing results to be about $348$ degrees,
which is quite different from all the values of $\sigma$ in
Figure~\ref{fig:1685a}.
We cannot really appreciate the effect of the singularity in the
evolution since we obtain very small values of the components
$\mathrm{Diff}_h(\frac{\partial\mathscr{H}}{\partial y_i})$. 

\begin{figure}[h]
  \centerline{\epsfig{figure=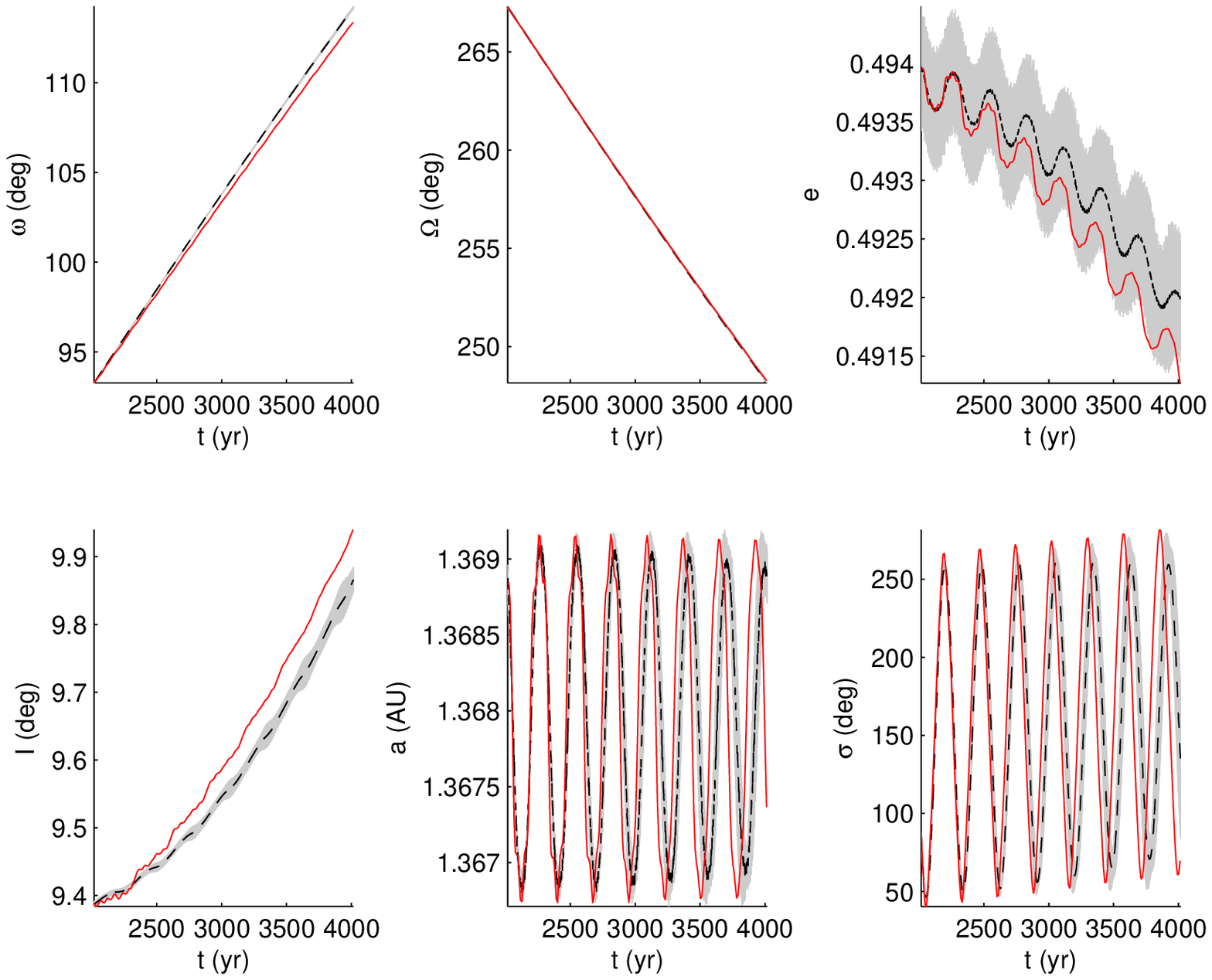,width=12cm}}
  \vskip 1cm
\centerline{\epsfig{figure=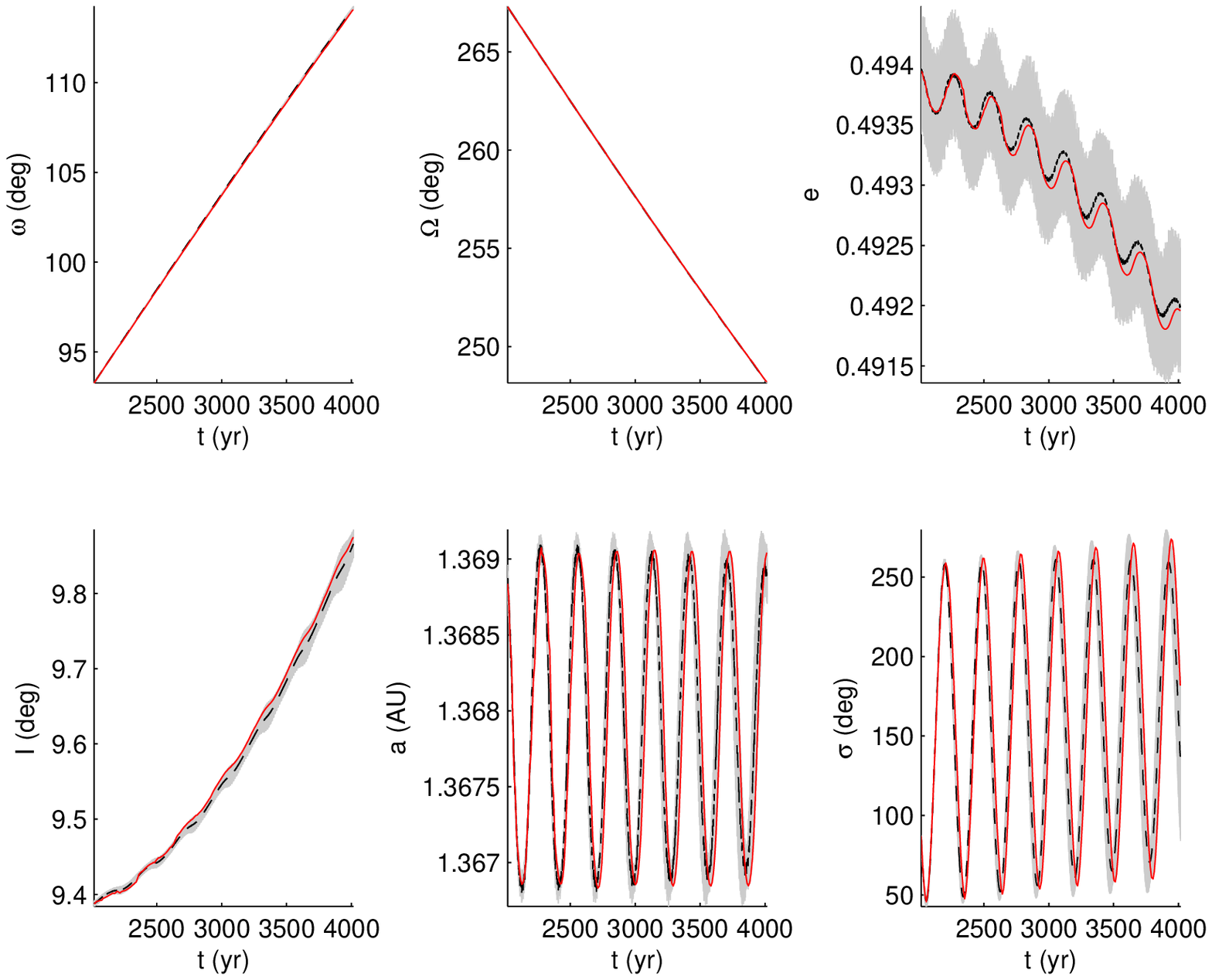,width=12cm}
  }
\caption{Asteroid 1685a: comparison between the long term evolution (solid line)
  and the arithmetic mean of 64 full numerical integrations (dashed line). Above $n_{\rm max}=3$. Below $n_{\rm max}=15$.}
\label{fig:1685a}
\end{figure}

\section{Conclusions}
We studied the long term dynamics of an asteroid under the
gravitational influence of the Sun and the solar system planets,
assuming that a mean motion resonance between the asteroid and one of
the planets occurs. We focused on the case of planet crossing
asteroids and considered a resonant normal form
$\mathscr{H}_{n_{\rm max}}$, see \eqref{normalform},\eqref{hfin}. The
analysis is performed separately for crossings with the resonant
planet or with another one. In both cases, we could define generalized
solutions of the differential equations for the long term dynamics,
going beyond the singularity. These solutions are continuous but in
general not differentiable. We also proved that generically, in a
neighborhood of a crossing time, the evolution of the signed orbit
distance along the generalized solutions is more regular that the long
term evolution of the orbital elements.
In case of crossings with the resonant planet, we recovered the
protection mechanism against collisions in the limit $n_{\rm max}\to
\infty$. This implies that, 
 if the resonant angle $\sigma$ is different from the
critical value $\sigma_c$ at the crossing times $t_c$ (see Sections~\ref{sec:gen},\ref{sec:dynpr}) also deep close
encounters are avoided, which makes the results of this theory more
reliable.  Indeed, close encounters can still occur with a planet not
involved in the resonance, and this represent a critical case.  Actually, in
this case the semimajor axis usually suffer a drastic change
\cite{opik}, pushing the asteroid outside the considered resonance.
By means of numerical experiments, in some relevant cases, we showed
that the model seems to approximate well the full evolution in a
statistical sense. We plan to make numerical tests
on a large scale, to study different dynamical behaviors of the
population of NEAs.

This work extends the results in \cite{gt13} to the resonant case and
gives a unified view of the orbit crossing singularity in case of mean
motions resonances with one planet: indeed, comparing the results in
Corollaries~\ref{cor1},\ref{cor2} we see how the discontinuity in the
derivatives, represented by
$\mathrm{Diff}_h\frac{\partial\mathscr{H}}{\partial y_i}$, vanishes in
a weak sense (i.e. in the sense of distributions) for $n_{\rm
  max}\to\infty$, if $\sigma\neq\sigma_c$.
Moreover, the resonant normal form introduced in \eqref{hfin} can
easily be extended to include more than one resonance, also with
different planets, by considering all the harmonics associated to the
corresponding resonant module (see \cite[Chap.2]{morby}).

\appendix

\section{Appendix}\label{appendix}
From the definition of the approximate distance $\delta_h$, we have that
\[
d^2(\E,V) = d^2_h(\E) + (V-V_h)\cdot\A_h (V-V_h) +{\cal R}_3^{(h)}(\E,V) =
\delta_h^2(\E) +{\cal R}_3^{(h)}(\E,V).
\]
We summarize below some relevant estimates and results from \cite{gt13}.
In the following, we shall denote by $c_i$, $i=1,\dots,14$, some
positive constants independent on $\E$. We first recall some Lemmas.
\begin{lemma}
  There exist positive constants $c_1$, $c_2$ and a neighborhood ${\cal U}$ of
  $(\E_c,V_h(\E_c))$ such that
  \begin{equation*}
c_1 \delta_h^2 \leq d^2 \leq c_2 \delta_h^2
\end{equation*}
holds for $(\cal E, V)$ in ${\cal U}$. Moreover, there exist positive
constants $c_3$, $c_4$ and a neighborhood ${\cal W}$ of $\E_c$ such that
\begin{equation}
d_h^2 + c_3\vert V-V_h\vert^2 \leq \delta_h^2 \leq
d_h^2 + c_4\vert V-V_h\vert^2
\label{aii}
\end{equation}
holds for $\E$ in ${\cal W}$ and for every $V\in\mathbb{T}^2$.
\label{inequalities}
\end{lemma}

\begin{lemma}
Using the coordinate change $\xi = {\cal A}_h^{1/2}(V-V_h)$ and then
polar coordinates $(\rho,\theta)$, defined by
$(\rho\cos\theta,\rho\sin\theta)=\xi$, we have
\begin{equation}
\int_{{\cal D}}\frac{1}{\delta_h}\,d\ell d\ell' = 
  \frac{1}{\sqrt{\det{\cal A}_h}}\int_{{\cal B}} 
  \frac{1}{\sqrt{d_h^2 + |\xi|^2}}\,d\xi
%
= \frac{2\pi}{\sqrt{\det{\cal A}_h}}(\sqrt{d_h^2+r^2} - d_h) \,,
\label{aiii}
\end{equation}
with ${\cal B} = \{\xi\in\mathbb{R}^2:|\xi|\leq r\}$.  The term $-2\pi
d_h/\sqrt{\det{\cal A}_h}$ is not differentiable at ${\cal E} =
{\cal E}_c\in\Sigma$.
\end{lemma}

\begin{lemma}
The maps
\begin{equation*}
\W^+\ni\E\mapsto \frac{\partial}{\partial  y_i}\int_{\T^2}\frac{1}{\delta_h(\E ,V )} dV, \quad \W^-\ni\E\mapsto \frac{\partial}{\partial  y_i}\int_{\T^2}\frac{1}{\delta_h(\E ,V )} dV 
\end{equation*}
can be extended to two different analytic maps $\G_h^+$, $\G_h^-$ such that, in $\W$,
 \[
\G_h^--\G_h^+ = 4\pi\biggl[\frac{\partial}{\partial
    y_i}\biggl(\frac{1}{\sqrt{\det(\mathcal{A}_h)}}\biggr)\tilde{d}_h+\frac{1}{\sqrt{\det(\mathcal{A}_h)}}\frac{\partial
    \tilde{d}_h}{\partial y_i}\biggr].
\]
\label{lem3}
\end{lemma}

Moreover the following estimates hold, with $\mathcal{U}_\Sigma =
\{(\E,V_h(\E)) : \E\in\Sigma\}$:
\begin{itemize}
\item[]
\begin{equation}
\int_{\D}\frac{\partial^2 }{\partial y_i\partial y_j}\frac{1}{d(\E,V)}dV\leq c_5 \quad\mbox{for } {\cal E} \mbox{ in } {\cal W},
\label{aiv}
\end{equation}
\item[]
 \begin{equation}
\left|\frac{\partial d^2}{\partial y_i}\right| , \left|\frac{\partial \delta_h^2}{\partial y_j}\right| \leq c_6(d_h+|V-V_h|)\quad \mbox{in }{\cal U}\setminus{\cal U}_\Sigma,
\label{av}
\end{equation}
\item[]
\begin{equation}
  \int_{\cal D}\frac{dV}{d_h +|V-V_h|}\leq c_7 \quad\mbox{for } {\cal E} \mbox{ in } {\cal W},
  \label{avi}
\end{equation}
\item[]
\begin{equation}
\left|\frac{\partial^2 \delta_h^2}{\partial y_i\partial y_j} \right|,\left|\frac{\partial^2 d^2}{\partial y_i\partial y_j} \right|\leq c_8 \quad\mbox{for } {\cal E} \mbox{ in } {\cal W},
\label{avii}
\end{equation}
\item[]
\begin{equation}
\left| \frac{1}{d^3}-\frac{1}{\delta_h^3} \right|\leq \frac{c_{9}}{d_h^2+|V-V_h|^2} \quad \mbox{in }{\cal U}\setminus{\cal U}_\Sigma,
\label{aviii}
\end{equation}
\item[]
\begin{equation}
 \frac{\partial{\cal R}_3^{(h)}}{\partial y_i}\leq c_{10}|V-V_h|^2
\quad \mbox{in }{\cal U}\setminus{\cal U}_\Sigma,
\label{aix}
\end{equation}
\item[]
\begin{equation}
\frac{\partial^2{\cal R}_3^{(h)}}{\partial y_i\partial y_j}\leq c_{11}|V-V_h|
\quad \mbox{in }{\cal U}\setminus{\cal U}_\Sigma,
\label{ax}
\end{equation}
\item[]
 \begin{equation}
\left| \frac{1}{d}-\frac{1}{\delta_h} \right|\leq {c_{12}} \quad \mbox{in }{\cal U}\setminus{\cal U}_\Sigma,
\label{axi}
\end{equation}
\item[]
 \begin{equation}
\left| \frac{\partial}{\partial y_i}\left(\frac{1}{d}-\frac{1}{\delta_h}\right) \right|\leq \frac{c_{13}}{d_h +|V-V_h|}
 \quad \mbox{in }{\cal U}\setminus{\cal U}_\Sigma,
\label{axii}
\end{equation}
\item[]
\begin{equation}
\frac{\partial V_h}{\partial  y_i}\leq c_{14} \quad\mbox{for } {\cal E} \mbox{ in } {\cal W}.
\label{axiii}
\end{equation}
\end{itemize}



\noindent {\bf Acknowledgments}.  We are grateful to Alessandro
Morbidelli, whose comments induced us to investigate better the
relation between this work and the other results present in the
literature, as explained in Section~\ref{sec:dynpr}.  The authors
acknowledge the support by the Marie Curie Initial Training Network
Stardust, FP7-PEOPLE-2012-ITN, Grant Agreement 317185. S.M. also
acknowledges financial support from the Spanish Ministry of Economy
and Competitiveness, through the ''Severo Ochoa'' Programme for
Centres of Excellence in R\&D (SEV-2015-0554), the project "Geometric
and numerical analysis of dynamical systems and applications to
mathematical physics" (MTM2016-76072-P), and the "Juan de la
Cierva-Formación" Programme (FJCI-2015-24917). G.F.G. has been
partially supported by the University of Pisa via grant PRA-2017
`Sistemi dinamici in analisi, geometria, logica e meccanica celeste'.


\newpage

\end{document}